\documentclass[11pt,a4paper]{article}

\usepackage{amsmath,amsthm,amssymb,amsfonts}

\usepackage{authblk}
\usepackage{setspace}
\usepackage{mathrsfs}
\usepackage{cases}
\usepackage{graphicx}
\usepackage{enumerate}
\usepackage{subcaption} 
\usepackage{datetime}

\numberwithin{equation}{section}
\newcommand{\ii}{{\rm i}}

\newcommand{\dd}{{\rm d}}

\newcommand{\x}{{\rm x}}
\newcommand{\y}{{\rm y}}
\newcommand{\z}{{\rm z}}

\newcommand{\vol}{{\rm vol}}

\newtheorem{thm}{Theorem}
\newtheorem{lemma}[thm]{Lemma}

\newtheorem{rem}[thm]{Remark}

\theoremstyle{definition}
\newtheorem{defn}[thm]{Definition}

\usepackage{color}

\newdateformat{daymonthyear}{\THEDAY \, \monthname[\THEMONTH] \THEYEAR}
\newdateformat{monthyear}{\monthname[\THEMONTH] \THEYEAR}

\begin{document}

\title{Semiclassical theories as initial value problems}
\author{Benito A. Ju\'arez-Aubry$^{1}$\thanks{{\tt benito.juarez@iimas.unam.mx}}, Tonatiuh Miramontes$^{2}$\thanks{{\tt tonatiuh.miramontes@correo.nucleares.unam.mx}} \\ and Daniel Sudarsky$^{2}$\thanks{\tt sudarsky@nucleares.unam.mx}}
\affil{$^{1}$Departamento de F\'isica Matem\'atica \\
Instituto de Investigaciones en Matem\'aticas Aplicadas y en Sistemas, Universidad Nacional Aut\'onoma de M\'exico,\\A. Postal 70-543, Mexico City 045010, Mexico}
\affil{$^{2}$Departamento de Gravitaci\'on y Teor\'ia de Campos \\Instituto de Ciencias Nucleares, Universidad Nacional Aut\'onoma de M\'exico, \\A. Postal 70-543, Mexico City 045010, Mexico}
\date{\daymonthyear\today}

\maketitle


\begin{abstract}
	
Motivated by the initial value problem in semiclassical gravity, we study the initial value problem of a system consisting of a quantum scalar field weakly interacting with a classical one. The quantum field obeys a Klein-Gordon equation with a potential proportional to the classical field. The classical field obeys an inhomogeneous Klein-Gordon equation sourced by the renormalised expectation value of the squared quantum field in a Hadamard state, $\langle \Psi| \Phi^2 \Psi \rangle$. Thus, the system of equations for the scalar fields reminisces of the semi-classical Einstein field equations with a Klein-Gordon field, where classical geometry is sourced by the renormalised stress-energy tensor of the quantum field, and the Klein-Gordon equation depends on the metric explicitly. We show that a unique asymptotic solution for the system can be obtained perturbatively at any fixed finite order in the weak coupling from initial data provided that the interaction is switched on and off smoothly in a spacetime region to the future of the initial data surface. This allows one to provide ``free" initial data for the decoupled system that guarantees that the Wightman function of the quantum field be of Hadamard form, and hence that the renormalised $\langle \Psi| \Phi^2 \Psi \rangle$ exist (in a perturbative sense) and be smooth. We comment on how to relax the switching of the interaction, which might be relevant for the corresponding problem in semiclassical gravity. 
\end{abstract}

\singlespacing

\section{Introduction}
\label{sec:Intro}

\subsection{Physical motivation and semiclassical gravity}

The interface between gravitation and quantum theory is arguably one of the most challenging   regimes to be studied in theoretical physics. The search for a truly unified understanding thereof, presumably described by a theory of quantum gravity, has been explored from various angles and, although some programmes have made notable progress, we do not yet have at our disposal any well established theory, known to be mathematically consistent, fully workable and able to deal with arbitrarily set problems on the subject. The approach known as semiclassical gravity,  in which the spacetime is treated in the classical language of general relativity, while matter is treated using the well-developed formalism of quantum filed theory in curved  spacetime is, on the other hand, often regraded as unsuitable  \cite{Page1981,  Eppley}, although this perspective has been  disputed  on  various grounds \cite{Carlip, Huggett, Mattingly:2005, Mattingly:2006}.

One should nevertheless be able to consider semiclassical gravity, if not as a fundamental theory, as a useful approximation with a limited range of applicability under suitable conditions. Among those conditions, one would expect the restriction that the scales of curvature involved be well below  those set by the Planck scale. Moreover, it seems natural to expect that if one is given the full theory of quantum gravity, the process of showing that under appropriate conditions one recovers general relativity would involve a succession of approximation steps that would at some point fall in a semiclassical gravity regime.

We should envision as an analogy to this attitude a theoretical analysis taking us, say from the physics of the standard model, expressed as a theory of quantum fields, to that of nuclear and atomic physics, expressed in the language of  non-relativistic quantum mechanics, to that of molecular dynamics and ending up in the hydro-dynamical description of fluids as characterised by something like  the Navier-Stokes equation. Even though we know very well that the latter equations are not a fundamental description of matter, we know they can be trusted to provide a reasonable description of the behavior of fluids under appropriate conditions. We take the same view regarding semi-classical gravity.


While semiclassical gravity has been widely studied in the literature, at least since the late seventies \cite{Wald:1978pj}, and is indeed a central motivation for the study of quantum fields in curved spacetimes, its initial value formulation has not been widely investigated, mainly due to the difficulties that the renormalisation of the stress-energy tensor bring about, which include introducing terms to the field equation that do not fall in the standard hyperbolic form for which the classical theorems of well-posedness can be applied, such as Cauchy-Kowalewskaya's \cite{Cauchy, Kowalevsky1875} or Leray's theorems \cite{Leray} (see also \cite{WaldGR}).

Indeed, the classical well-possedness of the  initial value problem for the Einstein field equations coupled to certain matter fields has been long established, following the pioneering work by Choquet-Bruhat \cite{FouresBruhat:1952,ChoquetBruhat:1969,ChoquetBruhat:1971}. In the semiclassical case, this task raises new challenges already for a free Klein-Gordon field -- the renormalisation-induced terms that spoil the hyperbolic form of the equations may lead to so-called runaway solutions for certain initial data. We take the spirit that criteria for avoiding spurious solutions need to be provided. Those issues will be discussed in a forthcoming work by the authors \cite{JKMS}.

We should mention that in certain situations involving high levels of symmetry, such as cosmological spacetimes, Wald's fifth axiom \cite{Wald:1978pj, Wald:1977up} can be enforced through a judicious choice of ambiguous renormalisation terms, and an initial value formulation can be obtained for semiclassical gravity. In cosmology, this strategy has been successfully exploited in the works \cite{Dappiaggi:2008mm, Pinamonti:2013zba, Pinamonti:2013wya}. More recently, the issue of the existence of solutions has been treated quite generally in cosmology by bringing the original system of equations to an infinite tower of ``moment" equations, but the relation between the solutions obtained from solving such tower of equations and initial data is not straightforward \cite{Gottschalk:2018kqt}.

Our long-term aim is to study solutions to full semiclassical gravity as initial value problems, subject to some physical criteria that discards spurious solutions. The physical solutions should consist of \cite{DiezTejedor2011} a {\it Lorentzian spacetime}, $(M, g_{ab})$ (with good properties, e.g. smoothness), a quantum field theory constructed over it and a singled-out physical state. In terms of a concrete representation, it would consist of the construction of the {\it representation of the quantum field}, $\hat \Phi$, as an operator-valued distribution field acting on a Hilbert space, $\mathscr{H}$, and the identification of {\it the physical state} $\Psi \in \mathscr{H}$, which is such that
\begin{subequations}
\label{SemiEinstein}
\begin{align}
 G_{ab}[g] + \Lambda g_{ab} & = 8\pi G_{\rm N} \langle\Psi | \hat T ^{\rm ren}_{ab}[\hat \Phi] \Psi \rangle, \label{SemiEinstein1} \end{align}
 while
 \begin{align}
 P[g] \hat \Phi & = 0, \label{SemiEinstein2}
\end{align}
\end{subequations}
where $P[g]$ denotes a differential operator that depends on the metric, as well as other parameters such the mass and curvature coupling, which defines the field equation for $\hat \Phi$. For example, in the case of a free Klein-Gordon field, $P[g] = \Box_g - m^2 - \xi R$.
 
We should mention at this stage that a necessary condition for the smoothness of $g_{ab}$ is that the physical state vector, $\Psi$, have a two-point function satisfying the Hadamard condition, see. def. \ref{HadState} below. This condition determines the distributional singularities of the two-point function of the quantum field in a precise way that ensures that the renormalisation procedure for the stress-energy tensor can be carried out and leads to a renormalised stress-energy tension that is smooth.\footnote{We should mention that in \cite{Pinamonti:2013zba, Pinamonti:2013wya} this condition has been weakened and lower-regularity cosmologies have been obtained as solutions.}
 
In this paper, we wish to focus on the a generic  difficulty in the obtention of semiclassical solutions, and give perspectives on it. Namely, that one needs to suitably characterise the evolution of the matter state, the quantum field and the spacetime geometry from initial data in a consistent manner that involves treating renormalised quantities -- such as $\langle\Psi | \hat T ^{\rm ren}_{ab}[\hat \Phi] \Psi \rangle$ -- as one evolves initial distributional data for the system \eqref{SemiEinstein}.

In order to gain insight on this rather difficult matter, we study a simplified model that captures some of the main features of eq. \eqref{SemiEinstein}. This model consists of two weakly interacting scalar fields without any gauge symmetries, in which one of them will be playing the role of the ``spacetime metric" in semiclassical gravity and thus will be treated as a classical field. The second field, which will be playing the role of the ``matter field", will be treated at the quantum level. 

\subsection{The scalars model, strategy and results}

In this model, the evolution of the classical field is defined by the Klein-Gordon equation with a source given by the renormalised expectation value of the squared field operator, $\langle\psi | \hat \Phi^2 \psi \rangle$, inspired by the appearance of the renormalised stress-energy tensor on the right-hand side of eq. \eqref{SemiEinstein1}. The equation of motion of the quantum field is the Klein-Gordon equation with a potential proportional to the classical field, inspired by the dependence of the quantum field equation on the metric tensor, cf. \eqref{SemiEinstein2}. See eq. \eqref{ScalarQuant} below and compare with eq. \eqref{SemiEinstein} above.


A fundamental requirement for our problem is that the initial data for the state be such that a Hadamard state (cf. def. \ref{HadState}) can be defined throughout the spacetime, which imposes constraints on the distributional singularities allowed on the initial-data-surface two-point functions, consistent with the canonical commutation relations that need be satisfied by the quantum field. It will turn out that this condition is best 	studied, by considering the problem in terms of an evolution equation for the field two-point function, which suffices to resolve the semiclassical interaction with the classical field.

If the above state condition is met, one can consistently reconstruct the interacting classical field, and the Wightman function of the field perturbatively in the coupling parameter as asymptotic series, such that at every step of the perturbation a notion of the Hadamard condition holds, and hence the (perturbatively-defined) renormalised expectation value $\langle\Psi | \hat \Phi^2 \Psi \rangle$ exists and is smooth. Note that this perspective differs from the standard perturbative backreaction approach to semiclassical theories, where the staring assumption is that a classical background is given so that one might construct a quantum field theory over it, and then the effects of the quantum fields are treated as perturbations over such background. For example, in semi-classical gravity a background metric is given at leading order and the effects of the renormalised stress-energy tensor are corrections to such metric.

We shall see that a key element in this model, which has no counterpart in semiclassical gravity and helps satisfy the requirements that we just discussed for the initial data, is the possibility to  smoothly turn on the interaction between the fields in the chronological future of the initial data surface, in such a manner that at early times the fields are decoupled, and hence ruled by free dynamics. This provides a way to construct initial data for the coupled system from initial data for the decoupled system, and solve the coupled system for the interacting fields in a perturbative way that guarantees that the classical field be smooth and the quantum field be in a Hadamard state.

\subsection{Organisation of the paper} 
 
This paper is organised as follows. After a brief discussion on the notations adopted in this paper at the end of the current section, in sec. \ref{sec:scalar} we pose the model of the interacting scalars as an initial value problem. We then show in sec. \ref{subsec:Decouple} how to solve the field equations for the quantum (resp. classical) field, whenever the classical (resp. quantum) field is taken as external background, which will be relevant for the strategy that we shall adopt for solving the coupled system, where a perturbative expansion in the coupling parameter reduces the problem to a tower of decoupled systems. Sec. \ref{subsec:Perturb} contains the main result of this paper, which shows in a constructive way how to obtain the solutions to the coupled system as an asymptotic series in the coupling parameter. Our conclusions appear in sec. \ref{sec:Outlook}.
 
\subsection{Notation}
\label{sec:Notation}

By a spacetime, $(\mathcal{M},g)$, we mean a real four-dimensional, connected (Hausdoff, paracompact) $C^\infty$ differentiable manifold, $\mathcal{M}$, equipped with a smooth Lorentzian metric $g$. We restrict our interest to those spacetimes that are time-orientable and assume a choice of time-orientation has been made and we further restrict our spacetimes to be globally hyperbolic \cite{BernalSanchez1, BernalSanchez2}, as we are interested in the initial value formulation in semiclassical theories. Our metric, $g$, has signature $(-,+,+,+)$. For $S\subset \mathcal{M}$, $J^+(S)$ denotes the causal future and $J^-(S)$ the causal past of the subset $S \subset\mathcal{M}$.

We use abstract index-notation, where latin indeces, $a, b, c, \ldots, h$, indicate the covariant and contravariant rank of the tensorial objects. Greek indices, $\mu, \nu, \dots$, are reserved for $4$-dimensional coordinate components. Latin indexes starting at $i$, i.e. $i, j, k, \dots$, will be used for $3$-dimensional coordinate components.

Covariant differentiation is indicated by a semi-colon when notationally advantageous, i.e., for a tensor of rank $(k,l)$, we denote $\nabla_a T_{a_1 \ldots a_k}^{b_1 \ldots b_l} = T_{a_1 \ldots a_k ;a}^{b_1 \ldots b_l}$. 

We set $\hbar = 1$ and $c=1$ unless otherwise specified. We denote spacetime points by Roman characters ($\x, \y, \dots$). Complex conjugation is denoted by an overline. Concrete operators on Hilbert spaces are surmounted with carets, such as $\hat{A}, \hat{B}, \dots$, while elements of an abstract non-commutative algebra are caret-free.

We use standard distributional notation. $\mathcal{F}: C_0^\infty(M) \to \mathbb{C}$ is a distribution, i.e. $\mathcal{F} \in \mathscr{D}'(M)$, if it is a linear functional on the space of smooth functions of compact support on $M$, $C_0^\infty(M)$. When there exits a functional kernel of $\mathcal{F}$, it is a function $F\in L^1(M)$, such that one has the representation $\mathcal{F}(f) = \int_M \dd \vol \, F(\x)f(\x)$, where $\dd \vol$ is the spacetime volume element, which can be locally written as $\dd \vol = (-\det g_{\alpha \beta})^{1/2} d^4x$. Oftentimes, when a distribution has no such functional kernel, one can represent the distribution with the aid of a one-parameter family of $L^1(M)$ functions, $F_{\epsilon}$ with $\epsilon>0$, whose limit as $\epsilon \to 0^+$ (if it exists) is not $L^1(M)$, as $\mathcal{F}(f) = \lim_{\epsilon \to 0^+} \int_M \dd \vol \, F_{\epsilon}(\x)f(\x)$. We refer to such object as the distributional kernel of the distribution $\mathcal{F}$, and will not in general make the distinction between a functional and its distributional kernel. Similar expressions are available for bi-distributions, such as two-point functions. 

Green functions in classical field theory are denoted by $E^-$ and $E^+$ for the advanced and retarded fundamental Green operators respectively, and by $E = E^- - E^+$ we denote the advanced-minus-retarded fundamental solution. For a given quantum state $\Psi$, $\,^\Psi G^+$ denotes the Wightman (bi-solution) two-point function.

For this two-point function we will often use the ``integrate then take the limit" prescription, whereby 
\begin{align}
\,^\Psi G^+(f,g) & = \lim_{\epsilon \to 0^+} \int_{M\times M} \dd \vol (\x) \, \dd \vol (\y)\,^\Psi G^+_\epsilon(\x,\y)f(\x)g(\y)
\end{align}
is the integral representation of the Wightman function, where $\,^\Psi G^+_\epsilon$ is regular for $\epsilon > 0$.

Given an operator $K: C_0^\infty (\mathcal{M}) \to C^\infty(\mathcal{M})$, and a bi-function, say $f \in C^\infty(\mathcal{M}\times \mathcal{M})$, $K_{\y} f(\x,\y)$ indicates $K$ is acting on $f(\x,\cdot)$ taking $\x \in \mathcal{M}$ as a fixed point and then the resulting function is evaluated at point $\y \in \mathcal{M}$. Similar notation is used for operators acting on bi-distributions.

Finally, $O(\alpha)$ denotes a quantity for which $O(\alpha)/\alpha$ is bounded as $\alpha \to 0$. We say a quantity $Q(\alpha)$ is order-$n$ in $\alpha$ if $n$ is the maximum number such that $Q(\alpha)/\alpha^n$ is bounded as $\alpha \to 0$.

\section{The model of interacting scalar fields}
\label{sec:scalar}

Let $(\mathcal{M},g)$ be a globally-hyperbolic spacetime with compact Cauchy hypersurfaces. Define the Lagrangian density
\begin{equation}
\mathscr{L}(\psi, \phi) =-\frac{1}{2}(\nabla_a\psi\nabla^a\psi + M^2\psi^2)-\frac{1}{2}(\nabla_a\phi\nabla^a\phi + m^2\phi^2) - \lambda \psi\phi^2,
\label{L}
\end{equation}
where $\psi: \mathcal{M} \to \mathbb{R}$ and $\phi: \mathcal{M} \to \mathbb{R}$ are classical fields, $M$ and $m$ are mass parameters for the fields $\psi$ and $\phi$ respectively and $\lambda \in \mathbb{R}$ is a coupling constant. The variation of the action given by the spacetime integral of the Lagrangian density \eqref{L} yields the equations of motion
\begin{subequations}
	\label{ScalarClass}
	\begin{align}
	(\Box - M^2)\psi = \lambda\phi^2, \\
	(\Box - m^2 - 2\lambda\psi)\phi=0, \label{KGClass}
	\end{align}
\end{subequations}
which govern the classical dynamics of the fields $\psi$ and $\phi$, where $\Box = g^{ab}\nabla_a \nabla_b$. 

We now define a semiclassical model where $\psi$ is treated classically and $\phi$ quantised (assuming the spacetime metric as fixed and given), as the formal system
\begin{subequations}
	\label{ScalarQuant}
	\begin{align}
	(\Box - M^2)\psi = \lambda \langle \Psi | \hat \Phi^2 \Psi \rangle_{\text{ren}}, \label{SQ1} \\
	(\Box - m^2 - 2\lambda\psi)\hat \Phi=0, \label{SQ2}
	\end{align}
\end{subequations}
where one replaces the classical observable $\phi$ in eq. \eqref{ScalarClass} by $\hat \Phi$, a suitable operator-valued distribution acting on a Hilbert space. The state $\Psi \in \mathscr{D} \subset \mathscr{H}$ lies in the domain of $\hat \Phi$ (suitably semared), and we assume that it is Hadamard (we will define what a Hadamard state is below), such that the quantity $\langle \Psi | \hat \Phi^2 \Psi \rangle_{\text{ren}}$ on the right-hand side of eq. \eqref{SQ1} is a smooth function obtained from the renormalisation of the expectation value of the operator $\hat\Phi^2$. The system is supplemented with suitable initial data on $\Sigma \subset \mathcal{M}$, an initial value Cauchy surface with normal vector $n^a$. Further assume that the spacetime foliated by Cauchy hypersurfaces parametrised by a global time function $T$, and such that the initial $\Sigma$ corresponds to $T=0$. The initial data is given by\footnote{Note that the name $\pi$ suggests an identification with the cannonical conjugate momenta $\pi_c$. Indeed, they are related by $\pi_c = \sqrt{\det h_{ij}} \pi$.}
\begin{subequations}
	\label{ScalarQuantData}
	\begin{align}
	\text{Initial data for } \psi : \quad & \left\{
	\begin{array}{rl}
	\psi|_\Sigma = & \!\!\! \varsigma , \\
	\nabla_n \psi|_\Sigma = & \!\!\! \varpi,
	\label{psidata1}
	\end{array}
	\right. \\
	\text{Initial data for } \hat{\Phi}: \quad & \left\{
	\begin{array}{rl}
	\hat \Phi|_\Sigma = \hat \varphi, \\
	\nabla_n \hat \Phi|_\Sigma = \hat \pi, \\
	\Psi|_\Sigma = \Psi.
	\label{Phidata}
	\end{array}
	\right.
	\end{align}
\end{subequations}
The third initial condition in the data \eqref{Phidata} will be taking as indicating that we will be working in the Heisenberg picture, i.e., that the states do not evolve but the field operators do. We demand that the initial data \eqref{Phidata} satisfies the CCR algebra relations, ie, for $\underline{\x}, \underline{\y} \in \Sigma$, 
\begin{enumerate}[\hspace{0.3 \linewidth}(i)]
	\item $[\hat \varphi(\underline{\x}), \hat \pi(\underline{\y})] = \ii \delta(\underline{\x}, \underline{\y})/\det(h_{ij}(\underline{\x})) 1\!\!1$.
	\item $[\hat \varphi(\underline{\x}), \hat\varphi(\underline{\y})] = 0$.
	\item $[\hat \pi(\underline{\x}), \hat\pi(\underline{\y})] = 0$.
\end{enumerate}
where $h_{ij}$ is the induced metric on $\Sigma$.

With sufficient symmetries available, an approach to solve the system \eqref{ScalarQuant} with initial data \eqref{ScalarQuantData} would be to write a parametrisation of the function $\psi$ describing it in terms of few parameters, make mode expansions for the field $\hat \Phi$ subject to the wave equation \eqref{SQ2} for arbitrary value of such parameters and where positive and negative frequency modes are distinguished. In turn, this allows to construct a generic Hilbert Fock space $\mathscr{H}$ upon which $\hat \Phi$ acts densely. Then it is necessary to find a physical state $\Psi \in \mathscr{H}$, which enforces equation \eqref{SQ1} while adjusting the parameters in such a way that conditions \eqref{ScalarQuantData} are met. This method to deal with a system possessing similar self-referring features has been used successfully in the context of simple semiclassical gravity settings involving inflation \cite{DiezTejedor2011,Canate2018}, but the technique has obvious limitations.

In the absence of symmetries, when constructing a semiclassical solution several difficulties appear. The first one is of course the intrinsic non-linearity brought about by the self-referential aspects of the problem. In full semiclassical gravity, for instance, a known difficulty is the fact that the spacetime metric depends on the quantum field and its state and that the quantum field construction depends on the spacetime metric. This fact is represented in the present problem by the dependence of $\psi$ on the quantum field and its state and that the quantum field theoretic construction depends on $\psi$. The second one is that separation in negative and positive frequency modes is arbitrary in general curved spacetimes, where no canonical choice is available, leading to unitarily inequivalent constructions\footnote{However, Fell's theorem \cite{Fell1960} proves that even nonunitarily equivalent constructions yield indistinguishable finite precision predictions for suitable matching pair of states in any construction.}. A third issue is that finding the necessary physical state that satisfies the semiclassical equations in this interacting setting is non-trivial because only for a special class of states -- so-called Hadamard states -- the term $\langle \Psi | \hat \Phi^2 \Psi \rangle_{\text{ren}}$ would make sense and, moreover, it is necessary to carry on a renormalisation procedure for this expression while simultaneously solving the coupled problem.

An avenue to address these difficulties, while specifying at the same time the necessary information of the state $\Psi$, is to resort to the algebraic approach, where a state is defined by its $n$-point functions. In the case of a quasi-free state, it suffices to specify its two-point function,
\begin{equation}
\,^\Psi G^+(u,v)= \langle \Psi | \hat \Phi (u) \hat \Phi (v)\Psi \rangle, \label{WightmanState}
\end{equation}
where $u, v \in C_0^\infty(\mathcal{M})$. 

The integral kernel of the Wightman two-point function, denoted by $\,^\Psi G^+_{\epsilon}(\x,\x')$, is a bi-solution to equation \eqref{KGClass} (in the relevant $\epsilon \to 0^+$ limit). If the two-point function has the correct singularity structure, then, Hadamard renormalisation can be used to define $\langle \Psi | \hat \Phi^2 \Psi \rangle_{\text{ren}}$ with the point-splitting and Hadamard subtraction procedure, as we shall discuss below. 

Henceforth, we treat the problem as one in which one seeks to determine the classical field configuration and the quantum field Wightman function. This entails replacing the initial data \eqref{Phidata} with initial data for the two-point function and its derivatives, i.e., to provide the expectation values of the equal-time fields $\hat \varphi$ and $\hat \pi$. 

We trade the system \eqref{ScalarQuant} with initial data \eqref{ScalarQuantData} by
\begin{subequations}
	\label{ScalarQuant2pt}
	\begin{align}
	(\Box_\x - M^2)\psi(\x) = \lambda\langle \Psi | \hat \Phi^2 (\x) \Psi\rangle_{\text{ren}} , \label{SQ1-2pt} \\
	(\Box_\x - m^2 - 2\lambda\psi(\x))\,^\Psi G^+(\x,\x')=0, \label{SQ2-2pt1}\\
	(\Box_{\x'} - m^2 - 2\lambda\psi(\x'))\,^\Psi G^+(\x,\x')=0, \label{SQ2-2pt2}
	\end{align}
\end{subequations}
with initial data on $\Sigma$ given by
\begin{subequations}
	\label{ScalarQuant2ptData}
	\begin{align}
	& \left\{
	\begin{array}{rl}
	\psi|_\Sigma = & \!\!\! \varsigma, \\
	\nabla_n \psi|_\Sigma = & \!\!\! \varpi,
	\label{psidata}
	\end{array}
	\right. \\
	& \left\{
	\begin{array}{rl}
	\,^\Psi G^+|_\Sigma = \,^\Psi G^+_{\varphi \varphi}, \\
	(\nabla_n \otimes 1)\,^\Psi G^+|_\Sigma = \,^\Psi G^+_{\pi \varphi}, \\
	(1 \otimes \nabla_n)\,^\Psi G^+|_\Sigma = \,^\Psi G^+_{\varphi \pi}, \\
	(\nabla_n \otimes \nabla_n)\,^\Psi G^+|_\Sigma = \,^\Psi G^+_{\pi \pi}. 
	\label{G+data}
	\end{array}
	\right.
	\end{align}
\end{subequations}

The initial data prescribed in eq. \eqref{G+data} defines bi-distributions on the initial surface, $\,^\Psi G^+_{\varphi \varphi},\,^\Psi G^+_{\pi \pi} : C_0^\infty(\Sigma) \times C_0^\infty(\Sigma) \to \mathbb{R}$ and $\,^\Psi G^+_{\varphi \pi},\,^\Psi G^+_{\pi \varphi} : C_0^\infty(\Sigma) \times C_0^\infty(\Sigma) \to \mathbb{C}$, which are nothing but the expectation values of equal-time fields and momenta, e.g., for $F, G \in C_0^\infty(\Sigma)$, given a concrete representation for the field operators in a Hilbert space, $\,^\Psi G^+_{\varphi \varphi}(F,G) = \langle \Psi | \hat \varphi(F) \hat \varphi(G) \Psi \rangle$. We demand that the initial data satisfy the CCR: For $\underline{\x}, \underline{\y} \in \Sigma$, 
\begin{enumerate}[\hspace{0.2 \linewidth}(i)]
	\item $\,^\Psi G^+_{\varphi \pi}(\underline{\x}, \underline{\y}) - \,^\Psi G^+_{\pi \varphi}(\underline{\y}, \underline{\x}) = \ii \delta(\underline{\x}, \underline{\y})/\det (h_{\alpha\beta}(\underline{\x}))$, \label{ConditionGi}
	\item $\,^\Psi G^+_{\varphi \varphi}(\underline{\x}, \underline{\y}) - \,^\Psi G^+_{\varphi \varphi}(\underline{\y}, \underline{\x}) = 0$, \label{ConditionGii}
	\item $\,^\Psi G^+_{\pi \pi}(\underline{\x}, \underline{\y}) - \,^\Psi G^+_{\pi \pi}(\underline{\y}, \underline{\x}) = 0$. \label{ConditionGiii}
\end{enumerate}
We also require data to observe hermiticity in the form
\begin{equation}
	\,^\Psi G^+_{\pi \varphi}(\underline{\x}, \underline{\y}) = \overline{\,^\Psi G^+_{\varphi \pi}(\underline{\y},\underline{ \x})}. \label{DataHermit}
\end{equation}

Note that the way in which we have formulated the problem \eqref{ScalarQuant2pt}, subject to initial data \eqref{ScalarQuant2ptData}, does not single out a preferred Hilbert space where $\hat\Phi$ acts. It allows to solve for the potential $\psi$ and obtain a Wightman function corresponding to the state $\Psi$, cf. eq. \eqref{WightmanState}, for all points in the spacetime, from which the GNS construction can be used to concretely represent field operators.

\section{Initial value formulation for the decoupled subsystems}
\label{subsec:Decouple}

The strategy that we will follow to solve the coupled system \eqref{ScalarQuant2pt} will consist of devising a perturbative approach in the coupling parameter, such that at each order the equations for the classical field and the quantum field decouple. In this section, we show how the decoupled systems have, each on their own, a well posed initial value formulation. 
\subsection{Classical field with fixed external source} \label{sec:ClassicalIVP}

Let us begin by discussing the problem \eqref{SQ1-2pt} considering that the right hand side of eq. \eqref{SQ1-2pt} is a fixed external, smooth source, $J(\x)$, ie,
\begin{equation}
(\Box_\x - M^2)\psi(\x) = J(\x), \label{KGExternalSource}
\end{equation}
subject to initial data \eqref{psidata}, 
\begin{align}
\psi |_{\Sigma} &= \varsigma, &  \nabla_ n \psi |_{\Sigma} &= \varpi. 
\end{align}

This is the initial value problem for a quasi-linear, second order hyperbolic partial differential equation. The existence and uniqueness of solutions for this problem was first proven by Cauchy \cite{Cauchy} and Kowalevskaya \cite{Kowalevsky1875}  in the analytic case, while the smooth case was proven by Leray \cite{Leray}  (Theorem 10.1.3 in Wald \cite{WaldGR}), see also Hawking \& Ellis \cite[Sec. 7.4]{Hawking1973}. In particular, existence and uniqueness of the fundamental advanced and retarded Green operators follow from the fact $\Box - M^2$ is a normally-hyperbolic operator (See, e.g., Wald \cite[Theorem 4.1.2]{Wald:1995yp}). Furthermore, all spacelike-compact, smooth solutions can be explicitly obtained with the aid of Green function methods, as we can see by the following lemma:
\begin{lemma} \label{LemmaPsi}
	Let $E^+_M$ be the fundamental retarded Green operator for the differential operator $\Box - M^2$. A representation for the solution to eq. \eqref{SQ1-2pt} with smooth initial data \eqref{psidata} and a smooth source $J$ is given by
	\begin{equation}
	\psi =\mathcal{Q} + E_M^+ J, \label{Greenpsi}
	\end{equation}
	where $\mathcal{Q}$ is a $C^\infty$ solution to 
	\begin{equation}
	(\Box - M^2)\mathcal{Q} = 0,
	\end{equation}
	subject to initial data of compact support given on a Cauchy hypersurface $\Sigma$ by
	\begin{subequations}
		\begin{align}
		\mathcal{Q}|_{\Sigma} &= \varsigma - E_M^+J|_\Sigma, \\
		\nabla_n \mathcal{Q} |_{\Sigma} &= \varpi - \nabla_n E_M^+J|_\Sigma.
		\label{fIC}
		\end{align}
	\end{subequations}
\end{lemma}
\begin{proof}
	Existence and uniqueness of $\mathcal{Q}$ is guaranteed by  Leray's theorem. Aplication of operator $\Box - M^2$ on \eqref{Greenpsi} results in 
	\begin{equation}
	(\Box_\x - M^2) \psi = J,
	\end{equation}
	and therefore $\psi$ is solution to \eqref{KGExternalSource}. Evaluation of \eqref{Greenpsi} at $\Sigma$ yields
	\begin{subequations}
		\begin{align}
		\psi|_{\Sigma} &=  \varsigma,\\
		\nabla_n \psi|_{\Sigma} &= \varpi,
		\end{align}
	\end{subequations}
	and therefore $\psi$ satisfies the  initial data boundary conditions \eqref{psidata}.
\end{proof}
Note that Leray's theorem already implies this is the unique solution for \eqref{KGExternalSource} subject to initial data \eqref{psidata}.

We can further express $\mathcal{Q}$ in terms of the initial data by means of Green theorem, yielding (see Appendix \ref{sec:IVF}, eq. \eqref{UInitial}), 
\begin{align}
\mathcal{Q} (\x) = \int_{\Sigma} \, \dd\Sigma ({\underline{\y}}) \,  \Big[ &E_M(\x,{\underline{\y}}) (\varpi({\underline{\y}}) - (\nabla_n (E_M^+ J))({\underline{\y}})) \nonumber \\
&-  (\varsigma ({\underline{\y}})- (E_M^+ J)({\underline{\y}}))(\nabla_n)_{{\underline{\y}}} E_M(\x,{\underline{\y}})\Big], \label{ClassicalHom}
\end{align}
where $E_M^-$ is the fundamental advanced Green operator for $\Box - M^2$ and $E_M = E_M^+ - E_M^-$. Thus, 
\begin{align}
\psi(\x) = \int_{\Sigma} \, \dd\Sigma ({\underline{\y}}) \, \Big[ &E_M(\x,{\underline{\y}}) (\varpi({\underline{\y}}) - (\nabla_n (E_M^+ J))({\underline{\y}})) \nonumber \\
&-  (\varsigma ({\underline{\y}})- (E_M^+ J)({\underline{\y}}))(\nabla_n)_{{\underline{\y}}} E_M(\x,{\underline{\y}})\Big] +  E_M^+ J, \label{SolClassSource}
\end{align}
or in condensed notation,
\begin{equation}
\psi = E_M(\varpi - \Pi)  - ((1\otimes \nabla_n)E_M )(\varsigma - \mathcal  J) +  E_M^+ J, \label{SolClassSourceCompact}
\end{equation}
where $\Pi = (\nabla_n (E_M^+ J))|_{\Sigma}$ and $\mathcal  J=( E_M^+J)|_\Sigma$.

\subsection{Quantum field with fixed external potential}
We now focus on the subsystem of eq. \eqref{SQ2-2pt1} and \eqref{SQ2-2pt2} subject to initial data defined by eq. \eqref{G+data}. This problem is well posed whenever $\psi$ is taken to be a smooth, fixed external potential $\vartheta$. This is so because the differential operator $(\Box - m^2 - 2\lambda\vartheta)$ is of normally hyperbolic form, and hence it has distinguished advanced and retarded fundamental Green operators, $E^\pm_{\vartheta}$, uniquely defined by their support properties, just as we have seen in the previous section. From now on, we will denote the functional dependence on this potential by $[\vartheta]$ or by a right subscript in the case of the propagators $E^{\pm}_{\vartheta}$. 

Whenever the data set \eqref{G+data} corresponds to initial data for a Wightman function, the complete function $\,^\Psi G^+ [\vartheta]$ can be extended to all spacetime with the aid of the advanced-minus-retarded propagator $E_{\vartheta} =E^-_{\vartheta} - E^+_{\vartheta}$ in the following way. Let $f, g \in C_0^\infty(\mathcal{M})$, then 
\begin{align}
\,^\Psi G^+[\vartheta](f,g) & = \,^\Psi G^+_{\varphi \varphi}\Big((\nabla_n E_{\vartheta}f )|_\Sigma, (\nabla_n E_{\vartheta}g )|_\Sigma\Big) - \,^\Psi G^+_{\varphi \pi} \Big((\nabla_n E_{\vartheta}f)|_\Sigma, (E_{\vartheta}g)|_\Sigma\Big) \nonumber\\
&\quad - \,^\Psi G^+_{\pi \varphi} \Big((E_{\vartheta}f)|_\Sigma, (\nabla_n E_{\vartheta}g)|_\Sigma \Big) + \,^\Psi G^+_{\pi \pi} \Big((E_{\vartheta}f)|_\Sigma, (E_{\vartheta}g)|_\Sigma\Big).
\label{G+IVP}
\end{align}

The statement of eq. \eqref{G+IVP} is naturally in correspondence with the standard equivalence between smeared spacetime quantum fields and symplectically smeared quantum fields (and momenta) on a Cauchy surface in globally hyperbolic spacetimes (see, e.g. \cite[Lemma A.1]{Dimock}), i.e., for any $(\mathcal{M},g)$ globally hyperbolic and $\Sigma \subset \mathcal{M}$ a Cauchy surface, for $f \in C_0^\infty(M)$, the fact that $ \hat \Phi $ is a distributional solution to the field equations can be expressed by
\begin{equation}
\hat \Phi (f) = \hat \pi(E_{\vartheta}f |_\Sigma)- \hat \varphi(\nabla_n E_{\vartheta}f |_\Sigma),\label{SyplSmear}
\end{equation}
which is the quantum version of eq. \eqref{UInitial} .

Equation \eqref{G+IVP} is simply the two-point function generalisation of \eqref{SyplSmear} (for details check Appendix \ref{sec:IVF}), showing that the initial value problem for the Wightman function is well posed whenever $\vartheta$ is taken as an external potential.  This result can be rephrased in terms of the standard integral kernel notation for distributions if we consider the regularised two point functions $\,^\Psi G^+_{\varphi \varphi\, \epsilon}(\underline{\x},\underline{\y})$, $\,^\Psi G^+_{\varphi \pi\, \epsilon} (\underline{\x},\underline{\y})$, $\,^\Psi G^+_{\pi \varphi\, \epsilon}(\underline{\x},\underline{\y})$ and $\,^\Psi G^+_{\pi \pi\, \epsilon}(\underline{\x},\underline{\y})$, defined from the restriction of the distributional kernel of $\,^\Psi G^+$ to the hypersurface $\Sigma$. As an explicit example, in the case of Minkowski spacetime with vanishing $\vartheta$ a privileged notion of vacuum ${\Omega_M}$ can be defined, and for this state, the Wightman function restricted to a flat fixed inertial time Cauchy hypersurface $\Sigma$ reduces to \cite{Weinberg}
\begin{align}
	\,^{\Omega_M} G^+_{\varphi \varphi} (F,G) &= \lim\limits_{\epsilon\rightarrow 0^+} \int_{\Sigma\times\Sigma} d^3 \underline{\x}\, d^3 \underline{\y}\, F(\underline{\x})G(\underline{\y}) \left(\frac{m K_1 \left(m\sqrt{|\underline{\x}-\underline{\y}|^2+\epsilon^2}\right)}{4\pi^2 \sqrt{|\underline{\x}-\underline{\y}|^2+\epsilon^2}}\right), 
\end{align}
where we denote by an underlining of the points on $\Sigma$ and by $|\underline{\x}-\underline{\y}|$ the standard euclidean distance between the two points. Thus, the regularised distributional kernel $\,^{\Omega_M} G^+_{\varphi \varphi\,\epsilon}$ is explicitly given by
\begin{align}
	\,^{\Omega_M} G^+_{\varphi \varphi\,\epsilon} (\underline{\x},\underline{\y}) = \frac{m K_1 \left(m\sqrt{|\underline{\x}-\underline{\y}|^2+\epsilon^2}\right)}{4\pi^2 \sqrt{|\underline{\x}-\underline{\y}|^2+\epsilon^2}}.
	\label{Minkophiphi}
\end{align}

The remaining regularised distributional kernels can be similarly obtained. Analog expressions can be obtained for a different state $\Psi$ in Minkowski spacetime for the case with vanishing potential. In this case, $\,^{\Psi} G^+_{\varphi \varphi\,\epsilon}$ may differ from $\,^{\Omega_M} G^+_{\varphi \varphi\,\epsilon}$ by a smooth bi-function and still possess a good singular structure. 

As we have mentioned above, physically-acceptable Wightman functions cannot be allowed to have arbitrary singularity structure. To characterise the acceptable singularity structure of such two-point functions in general, globally hyperbolic spacetimes, we need to introduce the notion of a two-point function having {\it Hadamard form}. To this end, let us first present the definition of the Hadamard fundamental solution:
\begin{defn}
	\label{def:Hadamard}
	A Hadamard fundamental solution, $H_\ell$, is a (distributional) solution to the equation
	\begin{equation}
	P_\x H_\ell (\x,\x') = 0,
	\label{eqPH}
	\end{equation}
	where $P_\x$ is a second order, normally hyperbolic linear operator acting on the first argument, whose distributional kernel admits a short distance expansion of the form
	\begin{equation}
	H_{\ell \, \epsilon}(\x,\x') = \frac{1}{8\pi^2}\left[\frac{\Delta^{1/2}(\x,\x')}{\sigma_\epsilon(\x,\x')} + v(\x,\x') \ln \left(\sigma_\epsilon(\x,\x')/\ell^2 \right) + w^{\ell}(\x, \x')\right] \label{HadamardF}
	\end{equation}
	in a convex region $D\ni \x, \x'$, where $\sigma_\epsilon(\x,\x') = \sigma(\x, \x') + 2 \ii [T(\x)-T(\x')]\epsilon + \epsilon^2$ is the regularised half-squared distance along the unique geodesic going from $\x$ to $\x'$, $T \in C^{\infty}(\mathcal{M})$ is a time function,
	$\Delta$ the van Vleck-Morette determinant\footnote{We also have $\Delta(\x,\x') \in C^\infty (\Sigma \times \Sigma)$.} and $v$ and $w^\ell$ are smooth bi-functions determined as formal power series in $\sigma$ by the {\it Hadamard recursion relations} (derived from eq. \eqref{eqPH}) for the expansion coefficients, subject to an arbitrarily-specified, smooth initial bi-function, $w_0$, for the recursion relations of $w^\ell$ and where $\ell$ is an arbitrary length scale.
\end{defn}

We are ready to define a Hadamard state.

\begin{defn}
\label{HadState}
A state of the Klein-Gordon field is called a {Hadamard state} if its Wightman two-point function differs from the Hadamard fundamental solution of the Klein-Gordon equation by a smooth bi-function. The Wightman two-point function of this state is said to have Hadamard form. 
\end{defn}

For example, in Minkowski spacetime with vanishing potential, the Wightman two-point function $\,^{\Omega_M} G^+ (\x, \y)$ admits a short distance expansion:
\begin{align}
	\,^{\Omega_M} G^+ (\x, \y) &= \frac{1}{8\pi^2}\Bigg[\frac{1}{\sigma(\x,\y) }+\frac{m^2}{2} \log
	\left(\frac{\sigma(\x,\y) }{\ell ^2}\right)+\frac{m^2}{2} \left[\log \left(\frac{m^2 \ell ^2}{2}\right) +2 \gamma -1\right]\nonumber \\
	&\quad + \sigma(\x,\y) \frac{m^4 }{16} \left[2 \log \left(\frac{\sigma(\x,\y) }{\ell ^2}\right)+\log
	\left(\frac{m^4 \ell ^4}{4}\right)+4 \gamma
	-5\right] +O\left(\sigma^2(\x,\y)\right) \Bigg], \label{MinkoVacHad}
\end{align}
which is of the Hadamard form, and therefore, $\Omega_M$ is a Hadamard state. In this expression, $\gamma$ is the Euler-Mascheroni constant.

At this stage we should emphasise that, while the Hadamard fundamental solution is not a bi-solution in the same way that the Wightman function is for the Klein-Gordon equation, it fails to be so by a smooth bi-function \cite{Fulling:1981,Moretti:2014}.

We henceforth demand that the data defined by eq. \eqref{G+data}, with fixed (not necessarily vanishing) $\vartheta$ in a fixed globally-hyperbolic spacetime, be such that, at the level of distributional kernels,
\begin{subequations}
	\label{G+dataHad}
	\begin{align}
	\,^\Psi G^+_{\varphi \varphi} - H_{\ell } [\vartheta]|_\Sigma \in C^\infty(\Sigma \times \Sigma), \\
	\,^\Psi G^+_{\pi \varphi} - (\nabla_n \otimes 1) H_{\ell} [\vartheta]|_\Sigma \in C^\infty(\Sigma \times \Sigma), \\
	\,^\Psi G^+_{\varphi \pi} - (1 \otimes \nabla_n) H_{\ell} [\vartheta]|_\Sigma \in C^\infty(\Sigma \times \Sigma),\\
	\,^\Psi G^+_{\pi \pi} - (\nabla_n \otimes \nabla_n) H_{\ell} [\vartheta]|_\Sigma \in C^\infty(\Sigma \times \Sigma),
	\end{align}
\end{subequations}
where $ H_{\ell} [\vartheta]$ is a Hadamard fundamental solution for the Klein-Gordon operator $(\Box - m^2 - 2\lambda\vartheta)$.

The conditions ennumerated in \eqref{G+dataHad} state nothing but the fact that the singularity structure of the initial data for $\,^\Psi G^+[\vartheta]$ is the same as that of the Hadamard fundamental solution. The singularities of the Hadamard fundamental solution ``initial data" are characterised by its wave-front set, which is propageted by the Klein-Gordon equation according to the Propagation of singularities theorem of Duistermaat and H\"ormander \cite[sec. 6.1]{duistermaat1972}, yielding the micro-local spectrum condition of Radzikowski \cite{Radzikowski:1996pa} for $ H_\ell [\vartheta]$. It follows from the regularity conditions in eq. \eqref{G+dataHad} that the Wightman function, $\,^\Psi G^+$, will also satisfy the micro-local spectrum condition, therefore ensuring the state is Hadamard. 

In view of the discussion above, for a Hadamard state, it is possible to express the distributional kernel of a two-point function $\,^\Psi G^+(\x,\y)$ as
\begin{align}
\,^\Psi G^+_\epsilon[\vartheta](\x, \y) = & \int_{\Sigma \times \Sigma} \dd \Sigma(\underline{\x}') \dd \Sigma (\underline{\y}') \,^\Psi G^+_{\varphi\varphi\,\epsilon}(\underline{\x}',\underline{\y}') \left( (\nabla_{n})_{\underline{\x}'} E_{\vartheta}(\underline{\x}',\x)\right) \left( ( \nabla_{n})_{\underline{\y}'} E_{\vartheta}(\underline{\y}',\y)\right) \nonumber \\
& \quad - \int_{\Sigma \times \Sigma} \dd \Sigma(\underline{\x}') \dd \Sigma (\underline{\y}') \,^\Psi G^+_{\varphi\pi\,\epsilon}(\underline{\x}',\underline{\y}') \left( (\nabla_{n})_{\underline{\x}'} E_{\vartheta}(\underline{\x}',\x)\right) E_{\vartheta}(\underline{\y}',\y) \nonumber \\
& \quad - \int_{\Sigma \times \Sigma} \dd \Sigma(\underline{\x}') \dd \Sigma (\underline{\y}') \,^\Psi G^+_{\pi\varphi\,\epsilon}(\underline{\x}',\underline{\y}') E_{\vartheta}(\underline{\x}',\x) \left( (\nabla_{n})_{\underline{\y}'} E_{\vartheta}(\underline{\y}',\y)\right) \nonumber \\
& \quad + \int_{\Sigma \times \Sigma} \dd \Sigma(\underline{\x}') \dd \Sigma (\underline{\y}') \,^\Psi G^+_{\pi\pi\,\epsilon}(\underline{\x}',\underline{\y}') E_{\vartheta}(\underline{\x}',\x)E_{\vartheta}(\underline{\y}',\y),
\label{Gxy}
\end{align}
such that equation \eqref{G+IVP} reads
\begin{equation}
	\,^\Psi G^+[\vartheta](f,g) = \lim\limits_{\epsilon \rightarrow 0^+}\int_{M\times M} \dd \vol(\x) \dd \vol(\y) \, f(\x) g(\y) \,\,^\Psi G^+_\epsilon[\vartheta](\x, \y).
\end{equation}

We collect these observations in the following theorem:
\begin{thm}
	The problem defined by the PDE system consisting of eq. \eqref{SQ2-2pt1} and \eqref{SQ2-2pt2}, with a fixed $\psi = \vartheta \in C^\infty(\mathcal{M})$, subject to initial data defined by eq. \eqref{G+data} satisfying the conditions ennumerated in \eqref{G+dataHad}, has a unique bi-solution, $\,^\Psi G^+ [\vartheta]: C_0^\infty(\mathcal{M} \times \mathcal{M}) \to \mathbb{C}$, determined by eq. \eqref{Gxy}, which is the distributional kernel of a Wightman function defining a Hadamard state.
	\qed \label{Th2}
\end{thm}

Note that once the system consisting of eq. \eqref{SQ2-2pt1} and \eqref{SQ2-2pt2} is solved, it is possible to define $\langle \Psi | \hat \Phi^2 (\x)\Psi \rangle [\vartheta]$, by
\begin{equation}
	\langle \Psi | \hat \Phi^2 (\x)\Psi \rangle [\vartheta] = \lim_{\\x' \to \\x} \left(\,^\Psi G^+ [\vartheta] (\x,\x')- H_{\ell_0}^0[\vartheta](\x,\x') \right), \label{HadamardVac0}
\end{equation}
modulo renormalisation ambiguities, where $H_{\ell_0}^0 [\vartheta]$ is a particular Hadamard fundamental solution for which $w^{\ell_0} [\vartheta]$ has the form
\begin{equation}
	w^{\ell_0} [\vartheta] (\x,\x')= \sum_{n=0}^{\infty} \sigma^n (\x,\x') w_n^{\ell_0} [\vartheta](\x,\x'), \label{w0}
\end{equation}
with\footnote{This choice is made in order to match the value of $\,^{\Omega_M} G^+$ in the limit $\x'\rightarrow\x$.} $w^{\ell_0}_0= \frac{m^2}{2}\left[ \log \left(\frac{m^2 {\ell_0}^2}{2}\right) + 2 \gamma - 1 \right]$. This definition depends on the arbitrary renormalisation length-scale parameter $\ell_0$. It can be shown that expanding the Hadamard fundamental solution using a renormalisation length-scale $\ell$ instead of $\ell_0$ will transform \eqref{HadamardVac0} into
\begin{equation}
\langle \Psi | \hat \Phi^2 (\x)\Psi \rangle [\vartheta] = \alpha\left(m^2 + 2\lambda \vartheta - \frac{1}{6} R(\x)\right) + \lim_{\x' \to \x}\left( \,^\Psi G^+ [\vartheta] (\x,\x')- H_\ell^0[\vartheta](\x,\x')\right), 
\end{equation}
with $\alpha\in \mathbb{R}$ a constant that relates the renormalisation lengths $\ell$ and $\ell_0$ by
\begin{equation}
	\alpha = \frac{1}{4\pi^2} \ln \left(\frac{\ell_0}{\ell}\right).
\end{equation}

This renormalisation scale is a parameter that cannot be fixed on mathematical grounds. Quite the opposite: the most general characterisation of ambiguities mathematically admissible in the definition of regularised objects like $\langle \Psi | \hat \Phi^2 (\x)\Psi \rangle$ generally refered to as \textit{Wick polynomials}, has been given by Moretti and Khavkine \cite{Moretti:2014}, generalising the original work by Wald and Hollands \cite{Hollands:2001}. The main result is that ambiguities can be shown to be polynomial functions of the field parameters $(m^2,\vartheta)$, as well as scalars tensorially formed from $g^{ab}$, $R_{abcd}, \dots$ and their derivatives. We will express these ambiguities as $C[\vartheta, m^2, g^{ab}, R_{abcd}, \nabla_a R_{bcde},\dots]$, where $C$ scales as $L^2 C$ when it's arguments are rescaled as $\vartheta \mapsto L^2 \vartheta$, $m^2 \mapsto L^2 m^2$, $g^{ab}\mapsto L^2 g^{ab}$, $R_{abcd}\mapsto L^{-2} R_{abcd}$. It can be readily seen that in our case the terms that allow $C$ to satisfy this rescaling condition are of the form
\begin{equation}
	C(\x) = \beta_1 m^2 + \lambda \beta_2 \vartheta(\x) + \beta_3 R(\x),
\end{equation}
so that we have in general
\begin{align}
\langle \Psi | \hat \Phi^2 (\x)\Psi \rangle [\vartheta] &= \beta_1 m^2 + \lambda \beta_2 \vartheta(\x) + \beta_3 R(\x) \nonumber\\
&\quad + \lim_{\x' \to \x}\left( \,^\Psi G^+ [\vartheta] (\x,\x')- \breve H_\ell^0[\vartheta](\x,\x')\right), \label{HadamardVacPolAmb}
\end{align}
 These ambiguities should be determined by further physical requirements for the vacuum polarisation or directly from experiments in a realistic model. 
 
\subsection{Perturbative solutions for the decoupled systems}

One may be tempted to use theorem \ref{Th2} and equation \eqref{Greenpsi} to prove the full interacting system \eqref{ScalarQuant2pt} has a well posed initial value formulation in terms of causal propagators and the sources, by setting $\vartheta= \psi$ and $J= \lambda\langle \Psi | \hat \Phi^2 \Psi \rangle_{\text{ren}} $. However, an immediate obstruction appears: Given that $\psi$ now functionally depends also on $\,^\Psi G^+$, the system turns into a coupled nonlinear system for either $\,^\Psi G^+[\psi]$ or $\psi[\,^\Psi G^+]$, and thus, it is not possible to define Green operators for the coupled problem. If we consider, however, that $\lambda$ is a small coupling constant, the hope, that will be realised below in sec. \ref{subsec:Perturb} is that one obtains an order-by-order decoupling, turning the coupled system into a infinitely countable set of linear, normally hyperbolic equations, allowing us to recover the well posed initial value formulation in a perturbative sense.

If we get back to the decoupled system and expand the operator $E^{\pm}_{\vartheta}$ in the parameter $\lambda$, a perturbative series is obtained in terms of $E^\pm_0$, the free advanced and retarded Green functions, ie., the fundamental Green operators of the differential operator $P_0= \Box - m^2$. The resulting expansion is\footnote{See appendix \ref{sec:PertE} for details.}
\begin{equation}
E^\pm_{P} = E_0^\pm \left( \sum_{k = 0}^n (2 \lambda \vartheta E_0^\pm )^k \right) + O(\lambda^{n+1}).
\label{Eexp}
\end{equation}

Equation \eqref{Eexp} is written in a compact form using distributional notation, but notice that in its integral representation the order-$k$ term will contain $k$ integrals. We should mention that Dimock has studied in detail propagator expansions of the form of eq. \eqref{Eexp} in the context of a quantum scalar field coupled to an external gauge potential, including convergence of the series expansion, in flat spacetime. See \cite{Dimock:1979cm}, and in particular cf. eq. (2.5) therein. Inserting eq. \eqref{Eexp} into eq. \eqref{G+IVP}, or equivalently in eq. \eqref{Gxy}, one obtains a perturbative expansion of the Wightman function in the parameter $\lambda$  corresponding to the given initial data. 

One can similarly obtain a power series expansion of the quantum field, $\hat \Phi$, itself using eq. \eqref{SyplSmear} and \eqref{Eexp}, with $\hat \Phi$ satisfying commutation relations $[\hat \Phi (f), \hat \Phi(g)] = -\ii E_\psi (f,g)$ for $f,g \in C_0^\infty(\mathcal{M})$, while the zeroth order term in the expansion, the ``free field" $\hat \Phi_0$, is a quantum field on its own right, satisfying the wave equation $(\Box - m^2) \hat \Phi_0 = 0$ and the commutation relations $[\hat \Phi_0 (f), \hat \Phi_0(g)] = -\ii E_0(f,g)$, in addition to the linearity and hermiticity axioms.

\section{Perturbative solutions in $\lambda$ to the initial value problem of the coupled system}
\label{subsec:Perturb}

We now study the semiclassical interacting scalar model introduced in sec. \ref{sec:scalar}, consisting of the problem \eqref{ScalarQuant2pt} with initial data \eqref{ScalarQuant2ptData}. The main result is that this problem is well-posed in a perturbative sense in the coupling parameter. No claims on convergence of the asymptotic series solutions are made. For technical reasons, as mentioned in the introduction, we will set the interactions between the fields to be switched on and off by a smooth function of compact support, $\chi$, for a finite interval of the foliation time $T$. While $\chi$ vanishes, the evolution of the system is that of decoupled fields, and hence well-posed. The system becomes an interacting one whenever $\chi$ is switched on, and the purpose of this section is to solve the problem in the corresponding region of spacetime.

As we shall see below, if the coupling between the fields is weak, the switching function $\chi$ permits the construction of solutions order by order in  $\lambda$ for the Wightman function of the quantum field and for the classical field, in such a way that the two-point function is Hadamard order by order, in a sense that we make precise in def. \ref{def:PertHad} below.

The main result of this section is theorem \ref{thm:main}, where we show that the aforementioned perturbative construction of solutions for the coupled system is possible out of initial data for the free fields.

We then discuss in section \ref{subsub:weak} how to weaken the requirement that initial data is provided at a Cauchy surface that does not intersect the support of the switching function, by analyzing first why does initial data for free fields cannot be used  to define a Hadamard state when the hypersurface intersects the support of the switching function, and propose a procedure to construct valid initial data for the coupled system when interaction is always on.

\subsection{The main result}

Let us begin by introducing the following definition:

\begin{defn}
	\label{def:PertHad}
	Let $(\mathcal{M},g)$ be a globally hyperbolic spacetime, and let $P = P_0 + \alpha V$, $P_0 = \Box - m^2$, $m>0$, $V \in C^\infty(\mathcal{M})$ ($V \neq -m^2$) and $\alpha \in \mathbb{R}$ a perturbative parameter for the potential $V$, $|\alpha| \ll 1$. Let $\,^\Psi G^+_{P_0}$ be the bi-solution to $P_{0}\otimes 1 \, \,^\Psi G^+_{P_0} = 1 \otimes P_0 \, \,^\Psi G^+_{P_0} = 0$, such that $\,^\Psi G^+_{P_0}(\x, \x') = \overline{\,^\Psi G^+_{P_0} (\x', \x)}$, i.e., $\,^\Psi G^+_{P_0}$ is an abstract Wightman function distributional kernel, which moreover satisfies the Hadamard condition in the sense that if $H_{\ell,P_0}$ is a Hadamard fundamental solution for $P_0$, then $\,^\Psi G^+_{P_0} - H_{\ell, P_0} \in C^\infty(M)$. 
	
We say that $\,^\Psi G^+_P$ is an {\it order-$n$ Wightman function for the differential operator $P$ if for a fixed $n \in \mathbb{N}$, it admits a perturbative expansion
	\begin{equation}\label{ExpansionGP}
	\,^\Psi G^+_P = \,^\Psi G^+_{P_0} + \sum_{k = 1}^n \,^\Psi G^+_{P\, k} \alpha^k
	\end{equation}
	that satisfies 
	\begin{equation}
		(P\otimes 1) \, \,^\Psi G^+_P = (1 \otimes P) \, \,^\Psi G^+_P = O(\alpha^{n+1}),
	\end{equation}
	 and each of the bi-functions $\,^\Psi G^+_{P\, k}$ is symmetric and real-valued.}

We say that $H_{\ell, P}^n$ is a Hadamard fundamental solution of order $n$ for the differential operator $P$ if it has a short distance expansion of the form \eqref{HadamardF}, satisfies 
\begin{equation}
(P \otimes 1) H_{\ell, P}^n =  O(\alpha^{n+1}),
\end{equation}
and admits a perturbative expansion of the form
\begin{equation}
H_{\ell, P}^n = H_{\ell, P_0} + \sum_{k = 1}^n H_{\ell, k} \alpha^k.
\label{Hk}
\end{equation}
 We further say that $\,^\Psi G^+_P$ defines an {\it order-$n$ Hadamard state} $\Psi$ if there exists a Hadamard fundamental solution of order $n$ for the differential operator $P$, such that for every $k \in \{1, \dots, n \}$, $G_{P\, k} - H_{\ell, k} \in C^\infty(M)$.
\end{defn}

\begin{rem}
If $P_0$ is the massless Klein-Gordon operator, expansions \eqref{ExpansionGP} and \eqref{Hk} require additional logarithmic corrections in the perturbative parameter. In order to simplify our analysis, we restrict to the massive case, where such polynomial expansion holds. 
\end{rem}

The following lemma shows how to construct an order-$n$ Wightman function for the Klein-Gordon operator in a given perturbative parameter, defininig an order-$n$ Hadamard state, out of initial data satisfying the CCR.

\begin{lemma}
	Let $(\mathcal{M},g)$ be a globally hyperbolic spacetime with compact Cauchy hypersurfaces. Let the global time-function $T \in C^\infty(\mathcal{M})$ define a foliation by Cauchy hypersurfaces, and denote by $n$ the future directed unit normal to those. Let $\Sigma_i \subset \mathcal{M}$, $\Sigma_{\rm on} \subset \mathcal{M}$ and $\Sigma_f \subset \mathcal{M}$ be the Cauchy surface defined by $T(\x) = t_i$, $T(\x) = t_{\rm on}$ and $T(\x) = t_f$ respectively with $ t_i < t_{\rm on} < t_f$. Let $\Omega \subset \mathcal{M}$ be the compact set $\Omega = J^+(\Sigma) \cap J^-(\Sigma_f)$ and let $\chi \in C_0^\infty(\mathcal{M})$ have as support the closure of $J^+(\Sigma_{\rm on}) \cap J^-(\Sigma_f)$. Let $P_0$, $P$, $\,^\Psi G^+_{P_0}$ and $H_{\ell, P_0}$ be as in def. \ref{def:PertHad}, with $V = \chi U$, where $U \in C^\infty(M)$, hence $V \in C_0^\infty(M)$ with ${\rm supp}(V) = {\rm supp} (\chi)$. Suppose that in addition to the CCR on $\Sigma_i$, the initial data defining $G_{P_0}$ satisfies 
	\begin{subequations}
		\label{Lem+dataHad}
		\begin{align}
		\,^\Psi G^+_{\varphi \varphi} - H_{\ell, P_0}|_{\Sigma_i} \in C^\infty(\Sigma_i \times \Sigma_i), \\
		\,^\Psi G^+_{\pi \varphi} - (\nabla_n \otimes 1) H_{\ell, P_0}|_{\Sigma_i} \in C^\infty(\Sigma_i \times \Sigma_i), \\
		\,^\Psi G^+_{\varphi \pi} - (1 \otimes \nabla_n) H_{\ell, P_0}|_{\Sigma_i} \in C^\infty(\Sigma_i \times \Sigma_i),\\
		\,^\Psi G^+_{\pi \pi} - (\nabla_n \otimes \nabla_n) H_{\ell, P_0}|_{\Sigma_i} \in C^\infty(\Sigma_i \times \Sigma_i),
		\end{align}
	\end{subequations}
	in the sense that the $\epsilon$-regularised integral kernels defining the expressions on \ref{Lem+dataHad} by an ``integrate then take the limit" prescription are smooth as $\epsilon \to 0^+$ inside the integral. Then, one can construct 
	\begin{equation}
	\,^\Psi G^+_P = \,^\Psi G^+_{P_0} + \sum_{k = 1}^n \,^\Psi G^+_{P\, k} \alpha^k + O(\alpha^{n+1}),
	\label{LemGPsum}
	\end{equation}
	such that $(P \otimes 1) \,^\Psi G^+_P = (1 \otimes P) \,^\Psi G^+_P = O(\alpha^{n+1})$, which is an order-$n$ Wightman function for the differential operator $P$ and defines an order-$n$ Hadamard state in $\alpha$.
	\label{LemHadRec}
\end{lemma}
\begin{proof}
	The fundamental Green operators for $P$, $E^\pm$, can be approximated up to order $\lambda^n$ as in eq. \ref{Eexp} with the aid of the free fundamental Green operators, $E^\pm_0$, as
	\begin{equation}
	E^\pm = E^\pm_0 \left(\sum_{k = 0}^n (-1)^k (\alpha V E^\pm_0)^k \right) + O(\alpha^{n+1}).
	\label{LemEpmsum}
	\end{equation}
	
	Eq. \eqref{LemEpmsum}, in turn yields the expansion for $\,^\Psi G^+_P$ in eq. \eqref{LemGPsum}, with the aid of eq. \eqref{Gxy} by inserting the $\alpha$-expansion for $E = E^- - E^+$, with each term in the expansion on the right-hand side of eq. \eqref{LemGPsum} obtained in terms of the initial data satisfying \eqref{Lem+dataHad}. In a geodesically convex neighbourhood, the short-distance singularities of $\,^\Psi G^+_P$ are the same that those of the Hadamard fundamental solution $H_{\ell, P}^0$. Expanding $H_{\ell, P}^0$ order by order in $\alpha$, as in eq. \eqref{Hk}, the singular structure of $\,^\Psi G^+_P$ is handled order by order, such that for each order-$k$ in $\alpha$ in the expansion, with $k \in \{1, \ldots n\}$, satisfies $\,^\Psi G^+_{P_k} - H_{\ell, k} \in C^\infty(\mathcal{M})$, and hence the expansion given by eq. \eqref{LemGPsum} is an order-$n$ term Wightman function for the differential operator $P$ and defines an order-$n$ Hadamard state.
\end{proof}

We now solve the interacting scalars system. 

\begin{thm}
	Let $(\mathcal{M},g)$ be a globally hyperbolic spacetime with compact Cauchy hypersurfaces and let $R$ be the Ricci scalar curvature. Let $T\in C^\infty(\mathcal{M})$ be a time function defining the Cauchy surfaces $\Sigma$ (with future directed unit normal $n$), $\Sigma_{\rm on}$ and $\Sigma_f$ as in lemma \ref{LemHadRec}, with $\Sigma_i \in J^-(\Sigma_{\rm on})$ and $\Sigma_f \in J^+(\Sigma_{\rm on})$. Let $\Omega \subset M$ be, as in lemma \ref{LemHadRec}, the compact set $\Omega = J^+(\Sigma) \cap J^-(\Sigma_f)$. Let $\varsigma, \, \varpi \in C_0^\infty(\Sigma)$, and $\,^\Psi G^+_{\varphi \varphi}, \,^\Psi G^+_{\pi \pi} \, : \, C_0^\infty(\Sigma) \times C_0^\infty(\Sigma) \rightarrow \mathbb{R}$ and $\,^\Psi G^+_{\pi \varphi}, \,^\Psi G^+_{\varphi \pi} \, : \, C_0^\infty(\Sigma) \times C_0^\infty(\Sigma) \rightarrow \mathbb{C}$ initial data for the Wightman two-point function of a free Klein-Gordon field in a Hadamard state, ie., satisfying \textsc{CCR},
	\begin{enumerate}[\hspace{0.3 \linewidth}(i)]
		\item $\,^\Psi G^+_{\varphi \pi}(\underline{\x}, \underline{\y}) - \,^\Psi G^+_{\pi \varphi}(\underline{\y}, \underline{\x}) = \ii \delta(\underline{\x}, \underline{\y})/\det (h_{\alpha\beta}(\underline{x}))$,
		\item $\,^\Psi G^+_{\varphi \varphi}(\underline{\x}, \underline{\y}) - \,^\Psi G^+_{\varphi \varphi}(\underline{\y}, \underline{\x}) = 0$,
		\item $\,^\Psi G^+_{\pi \pi}(\underline{\x}, \underline{\y}) - \,^\Psi G^+_{\pi \pi}(\underline{\y}, \underline{\x}) = 0$,
	\end{enumerate}	
hermiticity,
\begin{equation}
\,^\Psi G^+_{\varphi \pi}(\underline{\x}, \underline{\y}) = \overline{\,^\Psi G^+_{\pi \varphi}(\underline{\y}, \underline{\x})},
\end{equation}
and 
	\begin{subequations}
		\label{ThmG+dataHad}
		\begin{align}
		\,^\Psi G^+_{\varphi \varphi} - H_{\ell}^0|_{\Sigma_i} \in C^\infty({\Sigma_i} \times {\Sigma_i}), \\
		\,^\Psi G^+_{\pi \varphi} - (\nabla_n \otimes 1) H_{\ell}^0|_{\Sigma_i} \in C^\infty({\Sigma_i} \times {\Sigma_i}), \\
		\,^\Psi G^+_{\varphi \pi} - (1 \otimes \nabla_n) H_{\ell}^0|_{\Sigma_i} \in C^\infty({\Sigma_i} \times {\Sigma_i}),\\
		\,^\Psi G^+_{\pi \pi} - (\nabla_n \otimes \nabla_n) H_{\ell}^0|_{\Sigma_i} \in C^\infty({\Sigma_i} \times {\Sigma_i}),
		\end{align}
	\end{subequations}
	where $H_{\ell}^0$ is the Hadamard fundamental solution for 
	\begin{equation}
	(\Box_\x - m^2 ) H_{\ell}^0 (\x,\y)= 0.
	\end{equation}
	with $w_0=\frac{m^2}{2}\left[ \log \left(\frac{m^2 {\ell}^2}{2}\right) + 2 \gamma - 1 \right]$ (cf. eq. (\ref{w0})). Let $\beta_1, \beta_2, \beta_3 \in \mathbb{R}$ be arbitrary parameters, $\lambda \in \mathbb{R}$, $\chi \in C_0^\infty(\mathcal{M})$ with its support equal to the closure of $J^+(\Sigma_{\rm on}) \cap J^-(\Sigma_{f})$ and $\Lambda = \lambda \chi$. 
	
	The problem defined on $\Omega$ by
	\begin{subequations}
		\label{ThmScalarQuant2pt}
		\begin{align}
		(\Box_\x - M^2)\psi(\x) = \Lambda\langle \Psi | \hat \Phi^2(\x) \Psi \rangle, \label{ThmSQ1-2pt} \\
		(\Box_\x - m^2 - 2\Lambda\psi(\x))\,^\Psi G^+(\x,\x')=0, \label{ThmSQ2-2pt1}\\
		(\Box_{\x'} - m^2 - 2\Lambda\psi(\x'))\,^\Psi G^+(\x,\x')=0, \label{ThmSQ2-2pt2}
		\end{align}
	\end{subequations}
where the left hand side of eq. \eqref{ThmSQ1-2pt} is defined as
\begin{equation}
 \langle \Psi | \hat \Phi^2(\x) \Psi \rangle = \lim_{\x' \to \x} \left( \,^\Psi G^+(\x, \x') - H_\ell (\x, \x')\right) + \beta_1 m^2 + \Lambda \beta_2 \psi(\x) + \beta_3 R(\x) , \label{ThmVacPol}
\end{equation}
with $H_{\ell}$ the Hadamard fundamental solution for 
\begin{equation}
(\Box_\x - m^2 -2\Lambda\psi(\x)) H_{\ell} (\x,\y)= 0,
\end{equation} 
	and subject to initial data 
	\begin{subequations}
		\label{ThmScalarQuant2ptData}
		\begin{align}
		& \left\{
		\begin{array}{rl}
		\psi|_{\Sigma_i} = & \!\!\! \varsigma, \\
		\nabla_n \psi|_{\Sigma_i} = & \!\!\! \varpi,
		\label{Thmpsidata}
		\end{array}
		\right. \\
		& \left\{
		\begin{array}{rl}
		\,^\Psi G^+|_{\Sigma_i} = \,^\Psi G^+_{\varphi \varphi} ,\\
		(\nabla_n \otimes 1)\,^\Psi G^+|_{\Sigma_i} = \,^\Psi G^+_{\pi \varphi},\\
		(1 \otimes \nabla_n)\,^\Psi G^+|_{\Sigma_i} = \,^\Psi G^+_{\varphi \pi},\\
		(\nabla_n \otimes \nabla_n)\,^\Psi G^+|_{\Sigma_i} = \,^\Psi G^+_{\pi \pi},
		\label{ThmG+data}
		\end{array}
		\right.
		\end{align}
	\end{subequations}
admits a unique perturbative, asymptotic solution in the parameter $\lambda$ (satisfying $|\lambda| \ll 1$) up to a fixed order-$n$, $n \in \mathbb{N}$, of the form
	\begin{subequations}\label{LambdaSeries}
		\begin{align}
		\psi &= \sum_{k = 0}^n \psi_k \lambda^k + O(\lambda^{n+1}),\label{PsiLambda0}\\
		\,^\Psi G^+ &= \sum_{k= 0}^n \,^\Psi G^+_k \lambda^k +O(\lambda^{n+1}),\label{OmegaLambda}
		\end{align}
	\end{subequations}
i.e., eq. \eqref{LambdaSeries} solve the system \eqref{ThmScalarQuant2pt} up to order-$n$ so that $\,^\Psi G^+$ is an order-$n$ Wightman function for $\Box - m^2 - 2\Lambda\psi$ corresponding to an order-$n$ Hadamard state in the sense of def. \ref{def:PertHad}. 
	\label{thm:main}
\end{thm}

\begin{proof}
	The proof is constructive and consists on the sequential calculation of the coefficients $\psi_k$ and $\,^\Psi G^+_k$, $k \in \{ 0, \ldots, n\}$, in eq. \eqref{LambdaSeries}, while ensuring that $\,^\Psi G^+$ is an order-$n$ Wightman function for $\Box - m^2 - 2\Lambda\psi$ defining an order-$n$ Hadamard state in the sense of def. \ref{def:PertHad}.
	
	The proof relies on a perturbative construction of the advanced and retarded fundamental Green operators, $E^\pm$, for the operator $\Box - m^2 - 2 \Lambda \psi$. Let 
	\begin{subequations}
		\begin{align}
		\psi &= \sum_{k = 0}^n \psi_k \lambda^k + O(\lambda^{n+1}),\label{PsiLambda}\\
		E^\pm &= \sum_{k= 0}^n E^\pm_k \lambda^k +O(\lambda^{n+1}).\label{EpmLambda}
		\end{align}
	\end{subequations}
	Inserting the above expansions into $(\Box_\x - m^2 - 2 \Lambda \psi(\x)) E^\pm(\x,\x') = \delta_g(\x, \x')$, one can see that the coefficients on the right-hand side of eq. \eqref{EpmLambda} obey the recursion relations
	\begin{subequations}
		\label{Erecursion}
		\begin{align}
		&(\Box_\x - m^2) E^\pm_0 (\x, \x') = \delta_g (\x, \x'), \\
		&(\Box_\x - m^2) E^\pm_k(\x,\x') = 2 \sum_{j = 0}^{k-1} \psi_{j}(\x) \chi(\x) E_{k-j-1}^\pm(\x, \x'), \hspace{1cm} k \in \{1, \ldots, n\}.
		\end{align}
	\end{subequations}
	
	Thus, $E^\pm_0$ are the decoupled fundamental Green operators, and 
	\begin{subequations}
		\begin{align}
		E^\pm & = E^\pm_0 + \sum_{k = 1}^n E^\pm_k \lambda^k + O(\lambda^{n+1}) \\
		E^\pm_k & = 2 \sum_{k= 1}^n \lambda^k E^\pm_0\left( \sum_{j = 0}^{k-1} \psi_{j} \chi E_{k-j-1}^\pm \right), \hspace{1cm} k \in \{1, \ldots, n\}.
		\label{EpmSolution}
		\end{align}
	\end{subequations}
	
	The order-$k$ contribution, $E^\pm_k$, of the expansion of $E^\pm$, in eq. \eqref{EpmSolution} is given in terms of all the lower order $E^\pm_j$ and $\psi_j$ for $j \in \{0, \ldots, k-1 \}$. In order to obtain the $\psi_k$ coefficients of eq. \eqref{PsiLambda}, we write the right-hand side of eq. \eqref{ThmSQ1-2pt} as a power series in $\lambda$. Expanding eq. \eqref{ThmVacPol} we have
	\begin{align}
	\langle \Psi | \hat \Phi^2(\x) \Psi \rangle & = \lim_{\x' \to \x} \left[\,^\Psi G^+_0 (\x, \x') - H_{\ell,0} (\x, \x')\right] + \beta_1 m^2 + \beta_3 R(\x) \nonumber \\
	&\quad + \sum_{k = 1}^n \left( \lim_{\x' \to \x} \left[\,^\Psi G^+_k(\x, \x') - H_{\ell,k} (\x, \x')\right] + \beta_2 \chi \psi_{k-1}(\x)\right)\lambda^k+ O(\lambda^{n+1}).
	\label{VacPolk}
	\end{align}
	Let us define $\mathfrak{W}_{\ell,k} (\x) = \lim_{\x' \to \x} \left[\,^\Psi G^+_k(\x, \x') - H_{\ell,k} (\x, \x')\right]$. Substituting expansion \eqref{VacPolk} on \eqref{ThmSQ1-2pt} yields the following recursive equations for the $\psi_k$ coefficients
	\begin{subequations}
		\label{psiLambdaRec}
		\begin{align}
		&(\Box_\x - M^2) \psi_0(\x) = 0, \label{psi0} \\
		&(\Box_\x - M^2) \psi_1(\x) = \mathfrak{W}_{\ell,0} (\x) + \beta_1 m^2 + \beta_3 R(\x), \label{psi1} \\
		&(\Box_\x - M^2) \psi_k(\x) = \mathfrak{W}_{\ell,k-1} (\x) + \beta_2 \chi (\x) \psi_{k-2}(\x), \hspace{1cm} k \in \{2, \ldots, n\}.\label{psik}
		\end{align}
	\end{subequations}
	
	The recursion relations given by eq. \eqref{psiLambdaRec} are such that every coefficient $\psi_k$ depends on the lower order coefficients $\,^\Psi G^+_{k-1}$, $H_{\ell,k-1}$ and $\psi_{k-2}$. The sources on the right-hand side of eq. \eqref{psi1}, \eqref{psik} should in general contain distributional singularities, cf. eq. \eqref{VacPolk}, but whenever $\,^\Psi G^+$ can be realised as an order-$n$ Hadamard state in the sense of def. \ref{def:PertHad}, these sources will be smooth. We shall see below that this is the case with the aid of lemma \ref{LemHadRec}. Eq. \eqref{psiLambdaRec}, subject to the inital data for each $\psi_k$,
	\begin{subequations}
	\label{fks}
	\begin{align}
	& \left\{
	\begin{array}{rl}
	\psi_0|_\Sigma & = \varsigma, \\
	(\nabla_n \psi_0)|_\Sigma & = \varpi,
	\end{array}
	\right. \\
	& \left\{
	\begin{array}{rll}
	\psi_k|_\Sigma & = 0, & k \in \{1, \ldots, n\},\\
	(\nabla_n \psi_k)|_\Sigma & = 0, & k \in \{1, \ldots, n\},
	\end{array}
	\right. 
	\end{align}
\end{subequations}
	 define then a set of initial value problems with a \textit{fixed} source, which can be solved by standard Green-function methods just as in section \ref{sec:ClassicalIVP}. Let $E^\pm_M$ be the fundamental Green operators for the operator $(\Box_\x - M^2)$, and $E_M = E^-_M - E^+_M$. Then, application of eq. \eqref{SolClassSource} for each $\psi_k$, yields
	 	\begin{subequations} \label{psikSol}
	\begin{align}
	\psi_0 (\x) = \int_{\Sigma} \, \dd\Sigma ({\underline{\y}}) \, \Big[ E_M(\x,{\underline{\y}}) \varpi({\underline{\y}}) -  \varsigma ({\underline{\y}}) (\nabla_n)_{{\underline{\y}}} E_M(\x,{\underline{\y}})\Big] , \label{SolPsi0}
	\end{align}
	for $k=0$ and 
	\begin{align}
	\psi_k (\x) = \int_{\Sigma} \, \dd\Sigma ({\underline{\y}}) \, \Big[  \mathcal{J}_k ({\underline{\y}})(\nabla_n)_{{\underline{\y}}} E_M(\x,{\underline{\y}}) -E_M(\x,{\underline{\y}}) \Pi_k({\underline{\y}})  \Big] +  E_M^+ J_k, \label{SolPsik}
	\end{align}
\end{subequations}
for $k \in \{1, \ldots, n\}$, where 
\begin{subequations}
\begin{align}
	J_1 &= \mathfrak{W}_{\ell,0} (\x) + \beta_1 m^2 + \beta_3 R(\x),\\
	J_k &= \mathfrak{W}_{\ell,k-1} (\x) + \beta_2 \chi (\x) \psi_{k-2}(\x), \hspace{1cm} k \in \{2, \ldots, n\},
\end{align}
\end{subequations}
and 
\begin{subequations}
	\begin{align}
		\mathcal{J}_k &=( E_M^+ J_k ) |_{\Sigma},\\
		\Pi_k &= (\nabla_n (E_M^+ J_k))|_{\Sigma}.
	\end{align}
\end{subequations}
	We now show that $\,^\Psi G^+$ is an order-$n$ Wightman function for $\Box - m^2 - 2\Lambda\psi$ defining an order-$n$ Hadamard state in the sense of def. \ref{def:PertHad}. The two-point function coefficients, $\,^\Psi G^+_k$, are obtained in terms of the initial data with the aid of the advanced-minus-retarded propagator, $E = E^- - E^+$, with $E^\pm$ given by eq. \eqref{EpmSolution} and using formula \eqref{Gxy}. One has that
	\begin{align}
	\,^\Psi G^+_\epsilon (\x, \y)& = \sum_{i=0}^{n} (\,^\Psi G^+_\epsilon)_i (\x, \y) \lambda^i + O(\lambda^{n+1}) \nonumber \\ 
	& = \sum_{i=0}^{n} \lambda^i \sum_{j=0}^i \lim_{\epsilon \to 0^+} \left[ \int_{\Sigma \times \Sigma} \dd \Sigma(\underline{\x}') \dd \Sigma (\underline{\y}') \,^\Psi G^+_{\varphi\varphi \,\epsilon}(\underline{\x}',\underline{\y}') E_j(\underline{\x}',\x)E_{i-j}(\underline{\y}',\y) \nonumber \right. \\
	& - \int_{\Sigma \times \Sigma} \dd \Sigma(\underline{\x}') \dd \Sigma (\underline{\y}') \,^\Psi G^+_{\varphi\pi \,\epsilon}(\underline{\x}',\underline{\y}') \left( (\nabla_{n})_{\underline{\x}'} E_j(\underline{\x}',\x)\right) E_{i-j}(\underline{\y}',\y) \nonumber \\
	& - \int_{\Sigma \times \Sigma} \dd \Sigma(\underline{\x}') \dd \Sigma (\underline{\y}') \,^\Psi G^+_{\pi\varphi \,\epsilon}(\underline{\x}',\underline{\y}') E_j(\underline{\x}',\x) \left( (\nabla_{n})_{\underline{\y}'} E_{i-j}(\underline{\y}',\y)\right) \nonumber \\
	& \left. + \int_{\Sigma \times \Sigma} \dd \Sigma(\underline{\x}') \dd \Sigma (\underline{\y}') \,^\Psi G^+_{\varphi\varphi \,\epsilon}(\underline{\x}',\underline{\y}') \left( (\nabla_{n})_{\underline{\x}'} E_j(\underline{\x}',\x)\right) \left((\nabla_{n})_{\underline{\y}'} E_{i-j}(\underline{\y}',\y)\right) \right] \nonumber \\
	& + O(\lambda^{n+1}).
	\label{GxyExp}
	\end{align}
	
	Eq. \eqref{GxyExp} reveals that the coefficient $\,^\Psi G^+_k$ of the expansion of $\,^\Psi G^+$ depends on the $E^\pm$ coefficients of order $E^\pm_i$, $i \in \{0, \ldots, k\}$, and hence only on the coefficients $\psi_j$, $ j \in \{0, \ldots, k-1\}$. Similarly, it can be shown\footnote{See appendix \ref{sec:HadamardPert}} that the coefficients $H_{\ell,k}$ contain a dependence on the functions $\psi_j$ with $j \in \{0, \ldots, k-1\}$, completing the necessary condition for a sequential resolution of the system \eqref{ThmScalarQuant2pt} subject to initial data \eqref{ThmScalarQuant2ptData}. Application of Lemma \ref{LemHadRec} ensures that $\,^\Psi G^+$ is an order-$n$ Wightman function for $\Box - m^2 - 2\Lambda\psi$ defining an order-$n$ Hadamard state in the sense of def. \ref{def:PertHad}. 
\end{proof}

Theorem \ref{thm:main}, implies that, indeed, eq. \eqref{ThmVacPol} can be interpreted as the expectation value of the $\Phi^2$ operator of a perturbatively-constructed quantum field interacting with the classical field $\psi$,
\begin{align}
\hat \Phi (f) & = \sum_{k = 0}^n \lambda^k \hat \Phi_k(f) + O(\lambda^{n+1}) \nonumber \\
& = \sum_{k = 0}^n \lambda^k \left[ \hat \varphi(\nabla_n E_k f |_\Sigma) - \hat \pi(E_k f |_\Sigma)\right] + O(\lambda^{n+1}). \label{hatPhi}
\end{align}
that satisfies the {\rm (i)} linearity and {\rm (ii)} hermiticity Klein-Gordon axioms, and perturbative versions of {\rm (iii)} the commutation relations up to order-$n$, $[\hat \Phi (f), \hat \Phi (g)] = - \ii E(f,g) + O(\lambda^{n+1})$, which imply the relations
\begin{equation}
\sum_{j = 0}^k [ \hat \Phi_j(f), \hat \Phi_{k-j}(g)] = E_k(f,g), \quad k \in \{0, \ldots, n\},
\end{equation}
of which the $k = 0$ relation are the decoupled commutation relations, and {\rm (iv)} the Klein-Gordon equation axiom up to order-$n$,
\begin{equation}
(\Box - m^2 - 2 \lambda \chi \psi)\hat \Phi = O(\lambda^{n+1}),
\end{equation}
which imply the relations (cf. eq. \eqref{Erecursion})
\begin{subequations}
	\label{Phirec}
	\begin{align}
	&(\Box - m^2) \hat \Phi_0 = 0, \label{PhiDeco} \\
	&(\Box - m^2) \hat \Phi_k = 2 \sum_{j = 0}^{k-1} \psi_{j} \chi \hat \Phi_{k-j-1}, \hspace{1cm} k \in \{1, \ldots, n\},
	\end{align}
\end{subequations}
of which eq. \eqref{PhiDeco} is the decoupled field equation.

\subsection{Weakening of hypotheses and obstructions to constructing a Hadamard state}
\label{subsub:weak} 

Suppose that we prescind of the switching function $\chi$, or equivalently, take $\xi=1$ for all $\x \in \mathcal{M}$. Then, initial data for a Wightman function of the decoupled problem that is of the Hadamard form, like in eq. \eqref{Lem+dataHad}, will not in general correspond to initial data for the coupled problem having the Hadamard form.

This can be seen by realising that the short-distance distributional singularities of $\,^\Psi G^+$ on the initial value Cauchy surface $\Sigma$ should match those of the complete Hadamard fundamental solution on said hypersurface, but the data on $\Sigma$ for the Hadamard parametrices of the decoupled and the coupled problems differ by non-regular terms. More precisely, considering the differential operators 
\begin{subequations}
	\begin{align}
		& P_{\mu} = \Box- \mu, \\
		& P_{m} = \Box - m^2, 
	\end{align}
\end{subequations}
where $\mu^2=m^2 + \lambda\psi \in C^\infty(\mathcal{M})$ can be thought as a varying mas. Consider the Hadamard fundamental solutions for these operators, satisfying
\begin{subequations}
	\begin{align}
	& P_{m,\x} H_\ell (\x,\x')= 0, \\
	& P_{\mu,\x} H_{\ell}^0 (\x,\x') = 0.
	\end{align}
\end{subequations}
In a geodesically convex neighbourhood $N \subset \Omega$, we have that 
\begin{equation}
H_{\ell} - H_{\ell}^0 = (1/(8\pi^2)) (v_\mu - v_{m}) \ln (\sigma_\epsilon/\ell^2) + S,
\label{Hsubtraction}
\end{equation}
where $v_\mu$, $v_m$ and $S \in C^\infty(N \times N)$, and $v_\mu$ and $v_{m}$ are computed from their covariant Taylor expansion, subject to Hadamard recursion relations. More precisely, we have that $v_\mu = \sum_{k= 0}^\infty v_{\mu,k} \sigma^k$ with
\begin{subequations}
	\begin{align}
	0 & = 2 v_{\mu,0} + 2 v_{\mu,0 \,;a} \sigma^{;a} - 2 v_{\mu,0} \Delta^{-1/2} (\Delta^{1/2})_{;a} \sigma^{;a} + P_\mu \Delta^{1/2}, \\
	0 & = 2(k+1)(k+2) v_{\mu, k+1} + 2(k+1) v_{\mu, k+1 \, ;a} \sigma^{;a} - 2(k+1)v_{\mu, k+1} \Delta^{-1/2} (\Delta^{1/2})_{;a} \sigma^{;a} \nonumber \\
	& \quad + P_\mu  v_{\mu,k}, \quad k \in \mathbb{N}_0,
	\end{align}
\end{subequations}
where semicolons indicate covariant differentiation, and similarly for $v_{m} = \sum_{k = 0}^\infty v_{m,k} \sigma^k$.

On the initial value Cauchy surface, $\Sigma$, we have that
\begin{align}
H_{\ell}^\mu|_\Sigma - H_{\ell}^m|_\Sigma & = \frac{1}{8\pi^2} (v_{\mu,0} - v_{m,0})|_\Sigma \ln (\sigma_\epsilon|_{\Sigma}/\ell^2) + O\left(\sigma^{1/2}_\epsilon|_{\Sigma} \ln (\sigma_\epsilon|_{\Sigma}/\ell^2)\right) + S|_\Sigma \nonumber \\
& = -\frac{1}{4\pi^2} \lambda   \, \varsigma \, \ln (\sigma_\epsilon|_{\Sigma}/\ell^2) + O\left(\sigma^{1/2}_\epsilon|_{\Sigma} \ln (\sigma_\epsilon|_{\Sigma}/\ell^2)\right) + S|_\Sigma,
\end{align}
which for arbitrary $\varsigma$ is singular in the limit $\x\to \x'$. 

Therefore, initial data for the Wightman function satisfying conditions \eqref{Lem+dataHad}, is not appropriate when interaction is switched on. We purpose two strategies to deal with this difficulty.

One way to address this obstacle is to sequentially solve the system and the Hadamard fundamental solution at each order in $\lambda$, while correcting the initial data for the decoupled problem. This is possible given that at each order-$k$ in $\lambda$ with $k>1$, the terms $H_{\ell, k}$ (cf. \eqref{Hk}) are computed in terms of functions of order $k-1$ in $\lambda$ (See Appendix \ref{sec:HadamardPert}). These functions contain the correct singular structure for the Wightman function of the coupled problem at order $k$, and therefore contain the singular correction required at each order $k$. Our prescription is to use the same renormalisation scale $\ell$ to compute the terms $H_{\ell, k}$ and add these to the initial data. 
	
	Note that further addition of smooth terms at each order originates initial data that is equally valid for the coupled problem as the prescription we have just provided. For example, by using a different renormalisation scale, $\tilde{\ell}$, correction terms $H_{\tilde{\ell}, k}$ will differ from $H_{\ell, k}$ by smooth terms, so data corrected either by $H_{\ell, k}$ or by $H_{\tilde{\ell}, k}$ can be considered equally valid as initial data for the coupled problem. In this sense, our prescription is simply a minimal choice for this correction.
	
	\begin{rem}
		At order-$0$, initial data requires no corrections provided $H_{\ell, 0}=H_{\ell}^0$.
	\end{rem} 
	
This prescription may be summarised as the following correction rules for the initial data
\begin{subequations} \label{DataCorrection}
	\begin{align}
	\,^\Psi G^+_{\varphi \varphi} & \mapsto  \,^\Psi \widetilde{G}^+_{\varphi \varphi} = \,^\Psi G^+_{\varphi \varphi} + \sum_{k = 1}^n \lambda^k H_{\ell,k}=\sum_{k=0}^n \lambda^k \,^\Psi \widetilde{G}^+_{\varphi \varphi,k}, \\
	\,^\Psi G^+_{\varphi \pi} & \mapsto  \,^\Psi \widetilde{G}^+_{\varphi \pi}  = \,^\Psi G^+_{\varphi \pi} + \sum_{k = 1}^n \lambda^k (1 \otimes \nabla_n) H_{\ell,k} =\sum_{k=0}^n \lambda^k \,^\Psi \widetilde{G}^+_{\varphi \pi,k}, \\
	\,^\Psi G^+_{\pi \varphi} & \mapsto  \,^\Psi \widetilde{G}^+_{\pi \varphi}= \,^\Psi G^+_{\pi \varphi} + \sum_{k = 1}^n \lambda^k (\nabla_n \otimes 1) H_{\ell,k} =\sum_{k=0}^n \lambda^k \,^\Psi \widetilde{G}^+_{\pi \varphi,k}, \\
	\,^\Psi G^+_{\pi \pi} & \mapsto  \,^\Psi \widetilde{G}^+_{\pi \pi} = \,^\Psi G^+_{\pi \pi} + \sum_{k = 1}^n \lambda^k (\nabla_n \otimes \nabla_n) H_{\ell,k}=\sum_{k=0}^n \lambda^k \,^\Psi \widetilde{G}^+_ {\pi \pi,k}.
	\end{align}
\end{subequations}

The system then can be solved following the constructive procedure outlined in the proof of Theorem \ref{thm:main}, with $\chi=1$ for all times, $\Sigma_i$ can be taken to be any Cauchy hypersurface $\Sigma$, and a relevant difference, namely, that the perturbative two-point function $\,^\Psi G^+$ must be computed taking into account that initial data contains now corrections at each order in $\lambda$, so that eq. \eqref{GxyExp} is modified resulting in
\begin{align}
\,^\Psi G^+_\epsilon (\x, \y)& = \sum_{i=0}^{n} (\,^\Psi G^+_\epsilon)_i (\x, \y) \lambda^i + O(\lambda^{n+1}) \nonumber \\ 
& = \sum_{i=0}^{n} \lambda^i \sum_{k=0}^{i} \sum_{j=0}^{i-k} \lim_{\epsilon \to 0^+} \left[ \int_{\Sigma \times \Sigma} \dd \Sigma(\underline{\x}') \dd \Sigma (\underline{\y}') \,^\Psi \widetilde{G}^+_{\varphi\varphi \,\epsilon,k}(\underline{\x}',\underline{\y}') E_j(\underline{\x}',\x)E_{i-k-j}(\underline{\y}',\y) \nonumber \right. \\
& - \int_{\Sigma \times \Sigma} \dd \Sigma(\underline{\x}') \dd \Sigma (\underline{\y}') \,^\Psi \widetilde{G}^+_{\varphi\pi \,\epsilon,k}(\underline{\x}',\underline{\y}') \left( (n^{a'} \nabla_{a'})_{\underline{\x}'} E_j(\underline{\x}',\x)\right) E_{i-k-j}(\underline{\y}',\y) \nonumber \\
& - \int_{\Sigma \times \Sigma} \dd \Sigma(\underline{\x}') \dd \Sigma (\underline{\y}') \,^\Psi \widetilde{G}^+_{\pi\varphi \,\epsilon,k}(\underline{\x}',\underline{\y}') E_j(\underline{\x}',\x) \left( (n^{b'} \nabla_{b'})_{\underline{\y}'} E_{i-k-j}(\underline{\y}',\y)\right) \nonumber \\
& \left. + \int_{\Sigma \times \Sigma} \dd \Sigma(\underline{\x}') \dd \Sigma (\underline{\y}') \,^\Psi \widetilde{G}^+_{\varphi\varphi \,\epsilon,k}(\underline{\x}',\underline{\y}') \left( (n^{a'} \nabla_{a'})_{\underline{\x}'} E_j(\underline{\x}',\x)\right) \left((n^{b'} \nabla_{b'})_{\underline{\y}'} E_{i-k-j}(\underline{\y}',\y)\right) \right] \nonumber \\
& + O(\lambda^{n+1}).
\label{GxyExpCorrected}
\end{align}
It still holds that each two-point function coefficient $\,^\Psi G^+_k$ depends at most on the function coefficients $\psi_j$ with $j\in \{0, \dots, k-1\}$ provided that the bi-function coefficients $H_{\ell,k}$ depend on  $\psi_j$ with $j\in \{0, \dots, k-1\}$, just like for the $E^\pm_k$ propagator coefficients.

A different approach is to solve the system of interacting scalars with the switching function and initial data for the decoupled problem, just like in Theorem \ref{thm:main}, and generate initial data already in the form of a perturbative series of order-$n$ in $\lambda$ by evolving the system up to a Cauchy hypersurface such that $\Sigma_{\chi} \cap \operatorname{supp} \chi \neq \emptyset$. Then, take the new initial data to be 
\begin{subequations}
	\begin{align}
	\widetilde{\varsigma} &=\sum_{k=0}^n \lambda^k \widetilde{\varsigma}_k =\psi|_{\Sigma_{\chi}}, \\
	\widetilde{\varpi} &=\sum_{k=0}^n \lambda^k \widetilde{\varpi}_k =\nabla_n\psi|_{\Sigma_{\chi}},  
	\end{align}
	\begin{align}
	\,^\Psi \widetilde{G}^+_{\varphi \varphi} &=\sum_{k=0}^n \lambda^k \,^\Psi \widetilde{G}^+_{\varphi \varphi,k}= \,^\Psi G^+|_{\Sigma_\chi}, \\
	\,^\Psi \widetilde{G}^+_{\varphi \pi}  &= \sum_{k=0}^n \lambda^k \,^\Psi \widetilde{G}^+_{\varphi \pi,k}= (\nabla_n \otimes 1)\,^\Psi G^+|_{\Sigma_\chi}, \\
	\,^\Psi \widetilde{G}^+_{\pi \varphi} &=\sum_{k=0}^n \lambda^k \,^\Psi \widetilde{G}^+_{\pi \varphi,k}=  (1 \otimes \nabla_n)\,^\Psi G^+|_{\Sigma_\chi}, \\
	\,^\Psi \widetilde{G}^+_{\pi \pi} &=\sum_{k=0}^n \lambda^k \,^\Psi \widetilde{G}^+_ {\pi \pi,k} = (\nabla_n \otimes \nabla_n)\,^\Psi G^+|_{\Sigma_\chi},
	\end{align}
\end{subequations}
and proceed to solve the system as in the proof of Theorem \ref{thm:main} considering $\chi=1$ everywhere, exchange eq. \eqref{GxyExp} for \eqref{GxyExpCorrected}, take
 	\begin{subequations}
 	\label{fksCorrected}
 	\begin{align}
 	& \left\{
 	\begin{array}{rl}
 	\psi_k|_\Sigma & = \widetilde{\varsigma}_k, \\
 	(\nabla_n \psi_k)|_\Sigma & = \widetilde{\varpi}_k,
 	\end{array}
 	\right. 
 	\end{align}
 \end{subequations}
 for $k\in \{ 0, 1, \dots n \}$ instead of \eqref{fks}, and
 	\begin{align}
 	\psi_k (\x) = \int_{\Sigma} \, \dd\Sigma ({\underline{\y}}) \, \Big[ E_M(\x,{\underline{\y}}) (\widetilde\varpi({\underline{\y}})-\Pi_k({\underline{\y}})) -  (\widetilde\varsigma ({\underline{\y}}) - \mathcal{J}_k({\underline{\y}}))(\nabla_n)_{{\underline{\y}}} E_M(\x,{\underline{\y}})\Big] +  E_M^+ J_k, \label{psikSolSwOn}
 	\end{align}
 for $k \in \{0, \ldots, n\}$ instead of \eqref{psikSol},  keeping the definitions for $J_k$ (with $J_0=0$), $\mathcal{J}_k$, $\Pi_k$, and $\mathfrak{W}_{\ell,k}$, with $\,^\Psi G^+_k$ given by \eqref{GxyExpCorrected}.
 
 This effectively allows to obtain a solution for a system that is always coupled, from a solution of a system that decouples outside of the support of some switching function $\chi$. 
 
 These two strategies show how it is possible to relax the requirement that initial data be given in a region where the system is decoupled, and how to generalise our constructive procedure of solutions for data given in the form of a perturbative series obtained from evaluating a solution that is already known to be a consistent semiclassical system at fixed order in the coupling.

\section{Conclusions}
\label{sec:Outlook}

In this work, we have studied a weakly-coupled semiclassical system as an initial value problem, with the aims of understanding general properties of the semiclassical gravity equations. The particular model that we have studied is that of a classical scalar field $\psi$ weakly coupled to a quantum scalar $\Phi$, and in which the renormalisation of the expectation value of $\Phi^2$ needs to be performed as one solves the dynamical system. 
While the toy model developed in this work does not attempt to directly represent any physical system, its formal structure contains key elements that are a reminiscence of semiclassical gravity or semiclassical electrodynamics.

We have shown that it is possible to complete the analysis as a perturbative problem in the weak coupling, thus obtaining a perturbative asymptotic series that defines the classical field, up to a given power in the coupling, as well as a perturbative series for the Wightman two-point function of the quantum field, which is Hadamard in a precise perturbative sense, as well as for the expectation value of $\Phi^2$.

A key point in our analysis was the introduction of a switching function for the interaction between the fields that allows one to prescribe initial data in the ``free" regime, before the interaction begins. This turns out to be an important element in constructing a two-point function of the quantum field that satisfies the Hadamard condition.

 Aiming to relax the need for a switching function, we have proposed that, in a perturbative sense, initial data corresponds only to a zeroth-order contribution to the ``real" interacting initial data.  Therefore, one has to correct the initial data at each order in the coupling parameter to be able to recover a consistent semiclassical configuration defining a fixed order Hadamard state.  

A lesson for semiclassical gravity coming from this  work is that, although the present results   are   only  perturbative (to any finite order), they represent hope that a similar non perturbative version might be valid and thus also that the semiclassical system might be well posed in the sense that the renormalisation might be carried out together with the evolution of spacetime and of the quantum fields that are being constructed, while the state of the matter fields remains of Hadamard form throughout spacetime.  In other words    that given   suitable  initial data  for  the spacetime metric  and   Hadamard  like  two-point functions  (consistent  with the theory's constraints) on an  initial value hypersurface, one can construct a semiclassical configuration.

In the context of semiclassical gravity, we must note that there does not seem to be any reasonable  sense in which the gravity-matter interaction might  be ``switched off" in the past of any initial data hypersurface, negating the possibility of using the same strategy followed here. We think, however, that in such contexts, it is possible to provide initial data that has already all the elements that are needed to construct a state of Hadamard form for the fully gravity-matter interacting theory. These issues will be the subject of future work.

\section{Acknowledgements}

The authors acknowledge Prof. Bernard S. Kay for suggesting this problem to us during a fruitful collaboration visit to UNAM in Mexico City, as well as for his careful reading of this manuscript, providing valuable comments and corrections. B.A.J.-A. is supported by a DGAPA-UNAM postdoctoral fellowship. T.M. is supported by a CONACYT PhD. grant No. 280721. D.S. acknowledges partial financial support from the grants CONACYT No. 101712, and PAPIIT-UNAM No. IG100316, Mexico, as well as sabbatical fellowships from PASPA-DGAPA-UNAM and from Fulbright-Garcia Robles- COMEXUS. 

\appendix

\section{Initial value formulation for two-point functions}\label{sec:IVF}

Assume $(\mathcal{M},g)$ is a globally hyperbolic spacetime satisfying our requirements in sec. \ref{sec:Notation}. For the linear hyperbolic operator $P: C_0^\infty(\mathcal{M})\rightarrow C^{\infty}(\mathcal{M})$ given by
\begin{equation}
P f = \Box f- U f,
\end{equation}
for some test function $f\in C_0^\infty (\mathcal{M})$ and a potential $U\in C^{\infty}(\mathcal{M})$, a unique two-point function can be obtained from initial two-point function data on a Cauchy surface $\Sigma \subset \mathcal{M}$. The starting point are the unique\footnote{This uniqueness is due to global hyperbolicity and the normally hyperbolic form of the operator $P$.} advanced and retarded propagator operators, $E^+$ and $E^-$, $E^{\pm}: C_0^\infty(\mathcal{M})\rightarrow C^{\infty}(\mathcal{M})$, such that
\begin{align}
P E^{\pm}f = E^{\pm} P f = f, \label{PE}
\end{align}
and for any $f\in C_0^\infty$, with support properties
\begin{equation}
\operatorname{supp} (E^{\pm} f) \subset J^{\pm } \operatorname{supp}(f). \label{SuppE}
\end{equation}
Note that in distributional terms, we might write
\begin{equation}
(E^{\pm} f)(\x) = \int_{\mathcal{M}} \dd \vol(\y), E^{\pm}(\x,\y) f(\y),
\end{equation}
and the biscalars $E^{\pm} (x,y)$ are such that
\begin{subequations} \label{PEDelta}
	\begin{align}
	P_x E^{\pm}(\x,\y) &=\delta(\x,\y) /(-\det g_{\mu\nu}(\y)), \\
	P_y E^{\pm}(\x,\y) &=\delta(\x,\y)/(-\det g_{\mu\nu}(\x)).
	\end{align}
\end{subequations}
Note that equation \eqref{PE} in distributional kernel notation reads
\begin{equation}
\int_\mathcal{M} \dd \vol(\y) \, P_\x [E^{\pm}(\x,\y) f(\y) ]= \int_\mathcal{M} \dd \vol(\y) \, E^{\pm}(\x,\y) (P f) (\y)= f(\x).
\end{equation}

Note also that we have the following relation under exchange of evaluation points for $E^+$ and $E^-$, 
\begin{equation}
E^-(\x,\y) = E^+(\y,\x), \label{SymE}
\end{equation}
that can be derived from the uniqueness of propagators with support properties given by \eqref{SuppE}.

For a compact region $D \subset \mathcal{M}$, and smooth (or sufficiently regular) functions $u,v \in C^\infty(M)$ defined on $D$, Green's identity states\footnote{Note that $ u P v - v P u = u \Box v - v \Box u$.},
\begin{equation}
\int_{D}\dd\vol(\y)\, \left[ u(\y) P_\y v(\y) - v(\y) P_\y u(\y) \right] = \int_{\partial D} \dd S(\underline{\y}) \,\left[u(\underline{\y}) (\nabla_{n} v)(\underline{\y}) -v(\underline{\y}) (\nabla_{n} u)(\underline{\y}) \right], \label{GreenTh}
\end{equation}
where $\nabla_n$ is the outward normal derivative on $\partial D$, the boundary of $D$. Following the argument in \cite[Lemma A.1]{Dimock}, if $u$ is a smooth solution to $P u = 0$, and we take $v = E^+ f$ for some $f\in C_0^{\infty}$ as well as $D=D^-=J^-(\Sigma)-\Sigma$, we have
\begin{align}
\int_{D^-} & \dd \vol(\y) \, u(\y) f(\y) = \int_{\Sigma} \dd \Sigma(\underline{\y}) \,\left[u(\underline{\y}) (\nabla_n [E^+ f])(\underline{\y}) - (E^+f)(\underline{\y}) (\nabla_n u)(\underline{\y}) \right] \nonumber\\
&= \int_{\Sigma} \dd \Sigma(\underline{\y}) \,\left[u(\underline{\y}) \nabla_{n\, \underline{\y}} \int_{\mathcal{M}} \dd \vol(\z) E^+(\underline{\y},\z) f(\z) - (\nabla_n u)(\underline{\y}) \int_{\mathcal{M}} \dd \vol(\z) E^+ (\underline{\y},\z) f(\z) \right]. \label{Dmin}
\end{align}
Note that $E^+f $ has no support at the past infinity boundary of $D$, therefore such boundary term is not present in this expression.

Setting $v = E^- f$ and $D^+=J^+(\Sigma)-\Sigma$, we get
\begin{align}
\int_{D^+} &\dd\vol (\y) \, u(\y) f(\y) \nonumber\\
&= - \int_{\Sigma} \dd \Sigma(\underline{\y}) \,\left[ u(\y) \nabla_{n\,\underline{\y}}\int_{\mathcal{M}} \dd\vol (\z) E^-(\underline{\y},\z) f(\z) - (\nabla_{n} u)(\underline{\y}) \int_{\mathcal{M}} \dd\vol (\z) E^- (\underline{\y},\z) f(\z) \right]. \label{Dplus}
\end{align}

Then, given the data $u_0 = u|_{\Sigma}, \dot u_0 = \nabla_n u |_{\Sigma}$, adding \eqref{Dmin} and \eqref{Dplus} yields
\begin{align}
u(f) &= \int_{\Sigma} \dd\Sigma(\underline{\y})\,\left[u(\underline{\y}) \nabla_{n\, \underline{\y}} - (\nabla_{n} u)(\underline{\y}) \right] \int_{\mathcal{M}} \dd\vol(\z) \left[E^+(\underline{\y},\z)-E^-(\underline{\y},\z)\right] f(\z). \label{UInitial}
\end{align}

 Let $F(\x,\x'): C_0^\infty( \mathcal{M}\times \mathcal{M})$ be a bi-solution to $P$, ie.,
\begin{align}
P_\x F(\x,\x')&=0, \label{PF0}\\
P_{\x'} F(\x,\x') &= 0.\label{PF1}
\end{align}
Assume $F$ is $L^1(\mathcal{M}\times\mathcal{M})$ so that a bi-distribution $F$ can be defined for any $f,g\in C_0^\infty(\mathcal{M})$ by
\begin{equation}
	F(f,g) \equiv \int_{\mathcal{M}\times \mathcal{M}} \dd\vol(\x) \, \dd\vol (\x') \, F(\x,\x') f(\x) g(\x'). \label{DistFfg}
\end{equation}
For each $g$, by means of eq. \eqref{UInitial} we have
\begin{equation}
F(f,g) = \int_{\Sigma} \dd\Sigma(\underline{\x})\,\left[F(\underline{\x},g) \nabla_{n\, \underline{\x}} - (\nabla_{n\, \underline{\x}} F(\underline{\x},g)) \right] \int_{\mathcal{M}} \dd\vol(\z) \left[E^+(\underline{\x},\z)-E^-(\underline{\x},\z)\right] f(\z). \label{Green1}
\end{equation}
Then, for every $\underline{\x}\in\Sigma$ we also have
\begin{equation}
F(\underline{\x},g) = \int_{\Sigma} \dd\Sigma(\underline{\y})\,\left[F(\underline{\x},\underline{\y}) \nabla_{n\, \underline{y}} - (\nabla_{n\, \underline{y}} F(\underline{\x},\underline{\y})) \right] \int_{\mathcal{M}} \dd\vol(\z') \left[E^+(\underline{\y},\z')-E^-(\underline{\y},\z')\right] g(\z'). \label{Green2}
\end{equation}
Substituting \eqref{Green2} in \eqref{Green1} we get the final result
\begin{align}
F(f,g) &= \int_{\mathcal{M}\times \mathcal{M}} \dd\vol(\z)\dd\vol(\z') f(\z)g(\z') \nonumber\\
&\quad \times \left \lbrace \int_{\Sigma\times\Sigma} \dd\Sigma(\underline{\x})\,\dd\Sigma(\underline{\y})\, F(\underline{\x},\underline{\y}) \left[\nabla_{n\, \underline{\x}}E(\underline{\x},\z)\right] \left[\nabla_{n\, \underline{y}} E(\underline{\y},\z')\right]\right. \nonumber\\
&\qquad - \int_{\Sigma\times\Sigma} \dd\Sigma(\underline{\x})\,\dd\Sigma(\underline{\y})\,(\nabla_{n\, \underline{y}} F(\underline{\x},\underline{\y})) \left[\nabla_{n\, \underline{\x}}E(\underline{\x},\z)\right] \left[E(\underline{\y},\z')\right]\nonumber\\
&\qquad- \int_{\Sigma\times\Sigma} \dd\Sigma(\underline{\x})\,\dd\Sigma(\underline{\y})\,(\nabla_{n\, \underline{\x}} F(\underline{\x},\underline{\y}) )\nabla_{n\, \underline{y}} \left[E(\underline{\x},\z)\right] \left[E(\underline{\y},\z')\right] \nonumber\\ 
&\qquad \left. + \int_{\Sigma\times\Sigma} \dd\Sigma(\underline{\x})\,\dd\Sigma(\underline{\y})\, (\nabla_{n\, \underline{\x}} \nabla_{n\, \underline{y}} F(\underline{\x},\underline{\y})) \left[E(\underline{\x},\z)\right] \left[E(\underline{\y},\z')\right] \right\rbrace,
 \label{Ffg}
\end{align}
where we have used the advanced minus retarded propagator, $E=E^+ - E^-$. The quantity within braces can be directly identified with $F(\x,\y)$ according to \eqref{DistFfg},
\begin{align}
F(x,y) &= \int_{\Sigma\times\Sigma} \dd\Sigma(\underline{\x})\,\dd\Sigma(\underline{\y})\, F(\underline{\x},\underline{\y}) \left[\nabla_{n\, \underline{\x}}E(\underline{\x},\z)\right] \left[\nabla_{n\, \underline{y}} E(\underline{\y},\z')\right]\nonumber\\
&\qquad - \int_{\Sigma\times\Sigma} \dd\Sigma(\underline{\x})\,\dd\Sigma(\underline{\y})\,(\nabla_{n\, \underline{\y}} F(\underline{\x},\underline{\y})) \left[\nabla_{n\, \underline{\x}}E(\underline{\x},\z)\right] \left[E(\underline{\y},\z')\right]\nonumber\\
&\qquad- \int_{\Sigma\times\Sigma} \dd\Sigma(\underline{\x})\,\dd\Sigma(\underline{\y})\,(\nabla_{n\, \underline{\x}} F(\underline{\x},\underline{\y}) )\nabla_{n\, \underline{\y}} \left[E(\underline{\x},\z)\right] \left[E(\underline{\y},\z')\right] \nonumber\\ 
&\qquad + \int_{\Sigma\times\Sigma} \dd\Sigma(\underline{\x})\,\dd\Sigma(\underline{\y})\, (\nabla_{n\, \underline{\x}} \nabla_{n\, \underline{\y}} F(\underline{\x},\underline{\y})) \left[E(\underline{\x},\z)\right] \left[E(\underline{\y},\z')\right] .
\label{Fxy}
\end{align}
This explicitly shows how to reconstruct the complete bi-solution $F(\x,\x')$ defined on $\mathcal{M}\times \mathcal{M}$ by \textit{initial data} given on $\Sigma\times\Sigma$, corresponding to the two point functions $L^1(\Sigma\times \Sigma)$ defined for $\underline{\x}, \underline{\y} \in \Sigma$ by
\begin{align}
	F_{00}(\underline{\x},\underline{\y}) &= F(\underline{\x},\underline{\y}),\\
	F_{01}(\underline{\x},\underline{\y}) &= \nabla_{n\, \underline{\y}}F(\underline{\x},\underline{\y}) = \lim_{\y \to \underline{\y}}[\nabla_{n\, \y}F(\underline{\x},\y)],\\
	F_{10}(\underline{\x},\underline{\y}) &= \nabla_{n\, \underline{\x}}F(\underline{\x},\underline{\y})= \lim_{\x \to \underline{\x}}[\nabla_{n\, \x}F(\x,\underline{\y})],\\
	F_{11}(\underline{\x},\underline{\y}) &=\nabla_{n\, \underline{\x}} \nabla_{n\, \underline{\y}}F(\underline{\x},\underline{\y})= \lim_{\stackrel{\x \to \underline{\x}}{\y \to \underline{\y}}}[\nabla_{n\, \y} \nabla_{n\, \y}F(\x,\y)].
\end{align}
These initial value two-point functions define bi-distributions on $C_0^\infty(\Sigma)\times C_0^\infty(\Sigma)$, which we denote by $F_{00}=F|_\Sigma=$, $F_{01}=(1\otimes \nabla_n )F|_{\Sigma}$, $F_{10}= (\nabla_n \otimes 1)F|_{\Sigma}$ and $F_{11}=(\nabla_n \otimes \nabla_n )F|_{\Sigma}$, in terms of which eq. \eqref{Ffg} has a more compact form by using distributional notation,
\begin{align}
	F(f,g) &= F_{00}(\nabla_n E f|_{\Sigma}, \nabla_n E g|_{\Sigma}) - F_{01}(\nabla_n E f|_{\Sigma}, E g|_{\Sigma}) \nonumber \\
	&\quad - F_{10}(E f|_{\Sigma}, \nabla_n E g|_{\Sigma}) + F_{11}(E f|_{\Sigma}, E g|_{\Sigma}) .
\end{align}

This result holds for two point functions as long as these are regular and smooth; however, distributions can still be defined by the standard \textit{integrate then take the limit} prescription, for which, given regularised integral kernels $F_{00\,\epsilon}$, $F_{01\,\epsilon}$, $F_{10\,\epsilon}$ and $F_{11\,\epsilon}$, 
\begin{align}
	F(\x,\y) &= \lim_{\epsilon \to 0^+} \Bigg\lbrace \int_{\Sigma\times\Sigma} \dd\Sigma(\underline{\x})\,\dd\Sigma(\underline{\y})\, F_{00\,\epsilon}(\underline{\x},\underline{\y}) \left[\nabla_{n\, \underline{\x}}E(\underline{\x},\z)\right] \left[\nabla_{n\, \underline{y}} E(\underline{\y},\z')\right]\nonumber\\
	&\qquad - \int_{\Sigma\times\Sigma} \dd\Sigma(\underline{\x})\,\dd\Sigma(\underline{\y})\,F_{01\,\epsilon}(\underline{\x},\underline{\y})) \left[\nabla_{n\, \underline{\x}}E(\underline{\x},\z)\right] \left[E(\underline{\y},\z')\right]\nonumber\\
	&\qquad- \int_{\Sigma\times\Sigma} \dd\Sigma(\underline{\x})\,\dd\Sigma(\underline{\y})\,F_{10\,\epsilon}(\underline{\x},\underline{\y}) )\nabla_{n\, \underline{\y}} \left[E(\underline{\x},\z)\right] \left[E(\underline{\y},\z')\right] \nonumber\\ 
	&\qquad + \int_{\Sigma\times\Sigma} \dd\Sigma(\underline{\x})\,\dd\Sigma(\underline{\y})\,F_{11\,\epsilon}(\underline{\x},\underline{\y})) \left[E(\underline{\x},\z)\right] \left[E(\underline{\y},\z')\right] \Bigg\rbrace.
\end{align}

\section{Perturbative expansion of advanced and retarded propagators}
\label{sec:PertE}

Let $E_0^-$ and $E_0^+$ be the advanced and retarded propagators, respectively, for a certain hyperbolic operator $P_0$, over a spacetime with compact spatial section $(\mathcal{M},g)$. Let $\lambda>0$ and $V \in C^\infty(\mathcal{M})$ to be a scalar potential, and assume a second operator $P$ is defined as 
\begin{equation}
P = P_0 + \lambda V, \label{Ppert}
\end{equation}
such that $P$ is also hyperbolic, therefore, existence and unicity of the corresponding advanced and retarded propagators $E^-$ and $E^+$ for $P$ is granted. Assume $\lambda \ll 1$ is a perturbative parameter and that $E^\pm$ admits an expansion of the form
\begin{equation}
E^\pm = \sum_{n=0}^N \lambda^n \epsilon^{\pm}_n + O(\lambda^{n+1}), \label{ExpEPpert}
\end{equation}
up to a fixed order $N \in \mathbb{N}$, where each $\epsilon_n^\pm: \, C_0^\infty \rightarrow C^\infty$ satisfies the support property,
\begin{equation}
\operatorname{Supp}(\epsilon_n^\pm f) \subset J^{\pm} (\operatorname{Supp} \, f), \label{SuppEpsilon}
\end{equation} 
enforcing property \eqref{SuppE} for $E^{\pm}$ at each order in $\lambda$.

The assumption that an expansion of the form of eq. \eqref{ExpEPpert} will generally work, for example, if $P_0$ is the Plein-Gordon operator for a massive field and $V$ is a smooth potential. In the case of a massless field additional logarithmic terms are needed in the expansion.

The defining equation for $E^\pm$ is \eqref{PE}, which in this case reads
\begin{equation}
(P_0 + \lambda V) \left(\sum_{n=0}^N \lambda^n \epsilon^{\pm}_n + O(\lambda^{n+1}) \right) = 0,
\end{equation}	
\begin{equation}
P_0 \epsilon^{\pm}_0 + \sum_{n=1}^N \lambda^n (P_0 \epsilon^{\pm}_{n} + V \epsilon^{\pm}_{n-1}) + O(\lambda^{N+1}) = 0.
\end{equation}

Enforcing the above equation order by order yields the relations, for $\epsilon^{\pm}_n$ with $n>0$,
\begin{equation}
P_0 \epsilon^{\pm}_n + V \epsilon^{\pm}_{n-1} = 0,
\end{equation}
and 
\begin{equation}
P_0 \epsilon^{\pm}_0 = 0, \label{P0}
\end{equation}
along with the aforementioned support conditions. This conditions and \eqref{P0}, by means of the uniqueness of $E_0^\pm$ implies 
\begin{equation}
\epsilon^{\pm}_0 = E_0^{\pm}.
\end{equation}
By application of $E_0^{\pm}$ on the recursion relation for $n=1$, we get
\begin{equation}
\epsilon^{\pm}_1 = - E_0^{\pm} V E_0^{\pm},
\end{equation}
which in turn satisfies the required support condition for $\epsilon^{\pm}_1$, equation \eqref{SuppEpsilon}. Taking this procedure sequentially, it can be seen that
\begin{equation*}
\epsilon_n^{\pm} = (- E_0^{\pm} V )^n E_0. 
\end{equation*}
Therefore, the perturbative expansion for $E^{\pm}$ in terms of $E_0$ is
\begin{equation}
E^\pm = \sum_{n=0}^N \lambda^n (- E_0^{\pm} V )^n E_0 + O(\lambda^{n+1}). \label{ExpEPpertE0}
\end{equation}

Note that by property \eqref{PE}, we can also carry out this procedure by solving 
\begin{equation}
\left(\sum_{n=0}^N \lambda^n \epsilon^{\pm}_n + O(\lambda^{n+1}) \right) (P_0 + \lambda V) = 0,
\end{equation}
which yields the recursion relations
\begin{align}
\epsilon^{\pm}_n P_0 + \epsilon^{\pm}_{n-1} V= 0\,: n>0 ,\\
\epsilon^{\pm}_0 P_0 = 0, 
\end{align}
where application of $E_0^{\pm}$ on the right yields the corresponding result 
\begin{equation}
E^\pm =\sum_{n=0}^N \lambda^n E_0 (- V E_0^{\pm} )^n + O(\lambda^{n+1}). \label{ExpEPpertE0Alt}
\end{equation}

In terms of distributional kernels and a test function $f$, equations \eqref{ExpEPpertE0} and \eqref{ExpEPpertE0Alt} reads
\begin{align}
(Ef)(\x) &= \sum_{n=0}^N (-1)^n \lambda^n \int_{\mathcal{M}} \dd\vol(\y) \, E_0^{\pm}(\x,\y^{(n)}) V(\y^{(n)}) \int_{\mathcal{M}} \dd\vol(\y^{(n-1)})\, E_0^{\pm}(\y^{(n)},\y^{(n-1)}) V(\y^{(n-1)}) \nonumber \\
&\quad \int_{\mathcal{M}} \cdots \int_{\mathcal{M}} \dd\vol(\y^{(1)}) \, E_0^{\pm}(\y^{(2)},\y^{(1)}) V(\y^{(1)}) \int_{\mathcal{M}} \dd\vol(\y) \, E_0^{\pm}(\y^{(1)},\y) f(\y),\\
&= \sum_{n=0}^N (-1)^n \lambda^n \int_{\mathcal{M}} \dd\vol(\y) \, E_0^{\pm}(\x,\y^{(n)}) \int_{\mathcal{M}} \dd\vol(\y^{n-1}) \, V(\y^{(n-1)}) E_0^{\pm}(\y^{(n)},\y^{(n-1)}) \nonumber \\
&\quad \int_{\mathcal{M}} \cdots \int_{\mathcal{M}} \dd\vol(\y^{(1)}) \, V(\y^{(1)}) E_0^{\pm}(\y^{(2)},\y^{(1)})\int_{\mathcal{M}} \dd\vol(\y) \, V(\y) E_0^{\pm}(\y^{(1)},\y) f(\y) .
\end{align}

\section{Perturbative expansion of the Hadamard fundamental solution in weak coupling}\label{sec:HadamardPert}
In section \ref{subsec:Perturb} we have assumed the expansion coefficients $H_{\ell\, k}^0$ of the Hadamard fundamental solution shares the same functional dependence on $\psi_k$ than $G^+_k$. We will see this is the case according to the following lemma,
\begin{lemma}
	Let $P = \Box - m^2 - 2\lambda \psi $, $\psi \in C^\infty(\mathcal{M})$ and $\lambda \in \mathbb{R}$ a perturbative parameter for $\psi$, $|\lambda| \ll 1$, and assume $\psi$ itself depends on the parameter $\lambda$, such that the order-$n$ approximation of $\psi$ as a function of the perturbative parameter reads
	\begin{equation}
	\psi = \sum_{k=0}^{\tilde{n}} \lambda^k \psi_k + O(\lambda^{\tilde{n}}). \label{ExpandPsi}
	\end{equation}
	Let $H_\ell$ be the Hadamard fundamental solution for $P$, ie, a solution of the equation
	\begin{equation}
	(\Box_\x - m^2 - 2 \lambda \psi (\x)) H_{\ell,k} (\x,\x') = 0, \label{DefHpsi}
	\end{equation}
	of the form
	\begin{equation*}
	\hspace{0.08 \linewidth}H_\ell(\x,\x') = \frac{1}{8\pi^2}\left[\frac{\Delta^{1/2}(\x,\x')}{\sigma_\epsilon(\x,\x')} + v(\x,\x') \ln \left(\sigma_\epsilon(\x,\x')/\ell^2 \right) + w^{\ell}(\x, \x')\right], \hspace{0.07 \linewidth}\eqref{HadamardF}
	\end{equation*}
	just like in definition \ref{def:Hadamard}, such that for the regularisation parameter $\ell$, the smooth and regular function $w^{\ell}$ is order $\sigma$. Consider the order-$\tilde{n}$ approximation of $H_\ell$,
	\begin{equation}
	H_\ell = \sum_{k=0}^{\tilde{n}} \lambda^k H_{\ell\, k} + O(\lambda^{\tilde{n}+1}). \label{ExpandH}
	\end{equation} 
	Then, each term $ H_{\ell\, k}$ depends only on $\psi_{k-1}, \psi_{k-2}, \dots, \psi_0$.
\end{lemma}
\begin{proof}
	Substitution of the perturbative expansions \eqref{ExpandPsi} and \eqref{ExpandH} in \eqref{DefHpsi} demanding it to hold independently at each order in $\lambda$ yields the following system of equations,
	\begin{subequations} \label{HCoeffs}
		\begin{align}
		&(\Box_x - m^2) H_{\ell,0} (\x, \x') = 0, \\
		&(\Box_x - m^2) H_{\ell,k}(\x,\x') = 2 \sum_{j = 0}^{k-1} \psi_{j}(\x) H_{\ell,k-j-1}(\x,\x'), \hspace{1cm} k \in \{1, \ldots, n\},
		\end{align}
	\end{subequations}
	where just like in the case for $E^\pm$, at each order $k$, $H_{\ell,k}$ explicitly depends only on $\psi_{k-1}, \psi_{k-2}, \dots, \psi_0$. Nevertheless, a functional dependence on $\psi_k(\x')$ is not obviously discarded by the system \eqref{HCoeffs}. However, the complete set of Hadamard recursion relations allow for a constructive procedure where the sequential decoupling becomes explicit for the second entry.
	
	Consider the expansions
	\begin{subequations}\label{ExpansionsUVW}
		\begin{align}
		v(\x,\y) &= \sum_{n=0}^{N} v_n (\x,\y) \sigma^n(\x,\y) + O(\sigma^{N+1}),\label{ExpansionsUVWv}\\
		w_{\ell} (\x,\y) &= \sum_{n=0}^{N} w_n (\x,\y) \sigma^n(\x,\y) + O(\sigma^{N+1}). \label{ExpansionsUVWw}
		\end{align}
	\end{subequations}
	and substitute them in \eqref{DefHpsi} demanding it to hold independently at each order in $\sigma$. This yields
	\begin{equation}
	\Box_{\x} \sigma(\x,\y)= 4-2 \Delta^{-1/2}(\x,\y) \Big((\nabla_x)^a \sigma (\x,\y)\Big) \nabla_a \Delta^{1/2}(\x,\y),
	\end{equation}
	so that $\Delta^{1/2}$ can be identified as the Van-Vleck-Morette determinant,
	\begin{equation}
	\Delta^{1/2} (\x,\y)= - \frac{\operatorname{det}(-(\nabla_{\x})_\mu (\nabla_{\y})_\nu \sigma(\x,\y))}{\sqrt{-g(\x)}\sqrt{-g(\y)}},
	\end{equation}
	as well as the recursion relations for each $v_n $ and $w_n$,
	\begin{subequations}
		\begin{align}
		(1+[\nabla^a \sigma] \nabla_a -\triangle^{1/2} [\nabla_a \triangle^{1/2}] [\nabla_a \sigma]) v_0 &= -\frac{(\Box_x - m^2 - 2 \lambda \psi ) }{2} \triangle^{1/2}, \label{RecV0}\\
		(n+1+[\nabla^a \sigma]\nabla_a -\triangle^{-1/2}[\nabla^a \sigma][\nabla_a \triangle^{1/2}])v_{n} &= -\frac{(\Box_x - m^2 - 2 \lambda \psi )}{2n} v_{n-1}\, : n>0 , \label{RecVn}
		\end{align}
		\begin{align}
		\left( n + 2-\triangle^{-1/2}[\nabla^a\sigma ][\nabla_a \triangle^{1/2}]+ \nabla^a \sigma \nabla_a \right) w_{n+1} &= \nonumber \\
		- \frac{1}{2(n+1)}\Big[2 (2[n+1]- \triangle^{-1/2}[\nabla^a\sigma ][\nabla_a \triangle^{1/2}] &+ \nabla^a \sigma \nabla_a) v_{n+1} + (\Box_x - m^2 - 2 \lambda \psi ) w_n \Big], \label{RecW}
		\end{align}
	\end{subequations}
	where it has been used the fact that $\sigma(x,x')$ satisfies
	\begin{equation}
	2 \sigma (x,y) = g^{ab} \nabla_a \sigma (x,y) \nabla_b \sigma (x,y).
	\end{equation}
	
	Expanding now every $v_n$ and $w_n $ up to order order-$\tilde{n}$ in $\lambda$,
	\begin{subequations}
		\begin{align}
		v_n &= \sum_{k= 0}^{\tilde{n}}\lambda^k v_{n,i} + O(\lambda^{\tilde{n}+1}) ,\\
		w_n &= \sum_{k= 0}^{\tilde{n}} \lambda^k w_{n,i} + O(\lambda^{\tilde{n}+1}) ,
		\end{align}
	\end{subequations}
	and requiring the recursion relations to hold independently at every order in $\lambda$ yields the following system for the coefficients $v_{n,i}$,
	\begin{subequations} \label{RecVni}
		\begin{align}
		(1+[\nabla^a \sigma] \nabla_a -\triangle^{1/2} [\nabla_a \triangle^{1/2}] [\nabla_a \sigma]) v_{0,0} &= -\frac{(\Box_x - m^2) }{2} \triangle^{1/2}, \\
		(1+[\nabla^a \sigma] \nabla_a -\triangle^{1/2} [\nabla_a \triangle^{1/2}] [\nabla_a \sigma]) v_{0,k} &= \psi_{k-1} \triangle^{1/2},\\
		(n+1+[\nabla^a \sigma]\nabla_a -\triangle^{-1/2}[\nabla^a \sigma][\nabla_a \triangle^{1/2}])v_{n,0} &= -\frac{(\Box_x - m^2)}{2n} v_{n-1,0}, \\
		(n+1+[\nabla^a \sigma]\nabla_a -\triangle^{-1/2}[\nabla^a \sigma][\nabla_a \triangle^{1/2}])v_{n,k} &+ \frac{(\Box_x - m^2) v_{n-1,k} }{2n} = \nonumber \\
		&\quad \frac{1}{n} \sum_{j=0}^{k-1}\psi_{j}v_{n-1,k-j-1}, 
		\end{align}
	\end{subequations}
	where we can verify that every $v_{n,k}$ depends on $\psi_{k-1}, \psi_{k-2}, \dots , \psi_0$ only. For $w_{n,k}$, we have
	\begin{subequations} \label{RecWni}
		\begin{align}
		\left( n + 1-\triangle^{-1/2}[\nabla^a\sigma ][\nabla_a \triangle^{1/2}]+ \nabla^a \sigma \nabla_a \right) w_{n,0} &= \nonumber \\
		- \frac{1}{2 n}\Big[2 (2n- \triangle^{-1/2}[\nabla^a\sigma ][\nabla_a \triangle^{1/2}] &+ \nabla^a \sigma \nabla_a) v_{n,0} + (\Box_x - m^2 ) w_{n-1,0} \Big] \\
		\left( n + 1-\triangle^{-1/2}[\nabla^a\sigma ][\nabla_a \triangle^{1/2}]+ \nabla^a \sigma \nabla_a \right) w_{n,k} &\, \nonumber \\
		+ \frac{1}{2 n}\Big[2 (2n- \triangle^{-1/2}[\nabla^a\sigma ][\nabla_a \triangle^{1/2}] + \nabla^a \sigma \nabla_a) v_{n,k} &+ (\Box_x - m^2) w_{n-1,k}\Big] = \\
		&\qquad - \frac{1}{n}\Big[ \sum_{j=0}^{k-1} \psi_{j} w_{n-1,k-j-1} \Big] ,
		\end{align}
	\end{subequations}
	which again implies that at every order $k$ in $\lambda$, each $w_{n,k}$ depends explicitly only on $\psi_{k-1}, \psi_{k-2}, \dots, \psi_0$, granted that each $w_{0,k}$ satisfies this condition too. This is trivially the case when $w_0=0$. At this point we are ready to justify that no dependence on $\psi_k$ is present in either $v_{n,k}, w_{n,k}$. Expand $v_{n,k}, w_{n,k}$ as covariant Taylor series,
	\begin{subequations} \label{TaylorSeriesVW}
		\begin{align}
		v_{n,k} (\x,\y)= \sum_{p = 0}^{\infty} v_{n,k,a_1 \dots a_p}(x) \nabla^{a_1} \sigma(\x,\y) \cdots \nabla^{a_p} \sigma (\x,\y),\\
		w_{n,k} (\x,\y)= \sum_{p = 0}^{\infty} w_{n,k,a_1 \dots a_p}(x) \nabla^{a_1} \sigma(\x,\y) \cdots \nabla^{a_p} \sigma (\x,\y),
		\end{align}
	\end{subequations}
	where the ($p-tensor$) coefficients $v_{n,k,a_1 \dots a_p}(x)$ are given by
	\begin{subequations}
		\begin{align}
		v_{n,k,0}(\x) &= \lim\limits_{y \rightarrow x} v_{n,k}(\x,\y),\\
		v_{n,k,a}(\x) &= \lim\limits_{y \rightarrow x} (\nabla_a)_{\y} v_{n,k}(\x,\y),\\
		v_{n,k,a_1 a_2}(\x) &= \frac{1}{2} \lim\limits_{y \rightarrow x} (\nabla_{a_1} \nabla_{a_2})_{\y} v_{n,k}(\x,\y),\\
		&\vdots
		\end{align}
	\end{subequations}
	which can be read directly from \eqref{RecVni} after application of the corresponding number of derivatives and considering the properties of $\sigma$, $\nabla_a \sigma$, $\nabla_a \nabla_b \sigma$ and $\Delta$ in the limit $y \rightarrow x$. In this limit, equations \eqref{RecVni} turn into algebraic equations, discarding any possible implicit dependence on $\psi_{k}$ for any term $v_{n,k,a_1 \dots a_p}$. Therefore $v_{n,k,a_1 \dots a_p}(x)$ can only depend on $\psi_{k-1}, \dots, \psi_0$.
	
	Once $v_{n,k}[\psi_{k-1},\dots,\psi_0]$ is computed, the procedure can be done in an analogue manner for for each $ w_{n,k,a_1 \dots a_p}(x)$ by means of \eqref{RecWni}, obtaining $w_{n,k}[\psi_{k-1},\dots,\psi_0]$ . 
	
	Then, reconstructing $H_{\ell\, k}^0$ we have
	\begin{subequations}
		\begin{align}
		H_{\ell,0}^0 (\x,\y)&= \frac{1}{(4\pi)^2}\left[\frac{\Delta^{1/2}(\x,\y)}{\sigma (\x,\y)} + \sum_{n=0}^{N}\sigma^n(\x,\y) \left\lbrace v_{n,0} (\x,\y) \ln\left(\frac{\sigma(\x,\y)}{\ell^2}\right) + w_{n,0} (\x,\y) \right\rbrace \right]\nonumber \\
		&\qquad + O(\sigma^{M+1}) \\
		H_{\ell,k}^0 (\x,\y)&= \frac{1}{(4\pi)^2}\sum_{n=0}^{N}\sigma^n(\x,\y) \left[v_{n,k}[\psi_{k-1},\dots,\psi_0] (\x,\y) \ln\left(\frac{\sigma(\x,\y)}{\ell^2}\right) \right. \nonumber\\
		&\qquad \left. + w_{n,k}[\psi_{k-1},\dots,\psi_0] (\x,\y)\right] + O(\sigma^{M+1})
		\end{align}
	\end{subequations}
	where we have now explicitly expressed the functional dependence of every one of the coefficients, completing the proof.
\end{proof}

Note that $H_{\ell,0}^0$ is the Hadamard fundamental solution for the \textit{decoupled} problem with $P_0=\Box - m^2$, ie, $\lambda=0$, as expected. The correction coefficients $H_{\ell,k}^0$, with $k>0$, does not have the form of Hadamard fundamental solutions, but still include the relevant singularities dependent on $\psi_{k-1}, \dots, \psi_0$, of the full \textit{coupled} problem.


\end{document}